\UseRawInputEncoding
\documentclass[12pt,twoside]{article}

\usepackage{float}
\usepackage{graphicx}
\usepackage{epstopdf}
\usepackage{graphicx}
\usepackage{epic}
\usepackage{multirow}
\usepackage{pst-poly}  
\usepackage{pst-plot}  
\usepackage{pst-poly}  
\usepackage{tikz}
\usepackage{xcolor}
\usetikzlibrary{arrows,shapes,chains}

\renewcommand{\paragraph}{\roman{paragraph}}
\usepackage[a4paper]{geometry}
\makeatletter
\renewcommand\title[1]{\gdef\@title{\reset@font\Large\bfseries #1}}
\renewcommand\section{\@startsection {section}{1}{\z@}%
                                   {-3.5ex \@plus -1ex \@minus -.2ex}%
                                   {2.3ex \@plus.2ex}%
                                   {\normalfont\large\bfseries}}
\renewcommand\subsection{\@startsection{subsection}{2}{\z@}%
                                     {-3ex\@plus -1ex \@minus -.2ex}%
                                     {1.5ex \@plus .2ex}%
                                     {\normalfont\normalsize\bfseries}}
\renewcommand\subsubsection{\@startsection{subsubsection}{3}{\z@}%
                                     {-2.5ex\@plus -1ex \@minus -.2ex}%
                                     {1.5ex \@plus .2ex}%
                                     {\normalfont\normalsize\bfseries}}

\def\@runningauthor{}\newcommand{\runningauthor}[1]{\def\runningauthor{#1}}
\def\@runningtitle{}\newcommand{\runningtitle}[1]{\def\runningtitle{#1}}

\renewcommand{\ps@plain}{%
\renewcommand{\@evenhead}{\footnotesize\scshape \hfill\runningauthor\hfill}
\renewcommand{\@oddhead}{\footnotesize\scshape \hfill\runningtitle\hfill}}
\newcommand{\rmnum}[1]{\romannumeral #1}
\newcommand{\Rmnum}[1]{\expandafter\@slowromancap\romannumeral #1@}

\newcommand{\F}{\mathbb{F}}

\newcommand {\C}{{\mathcal{C}}}

\newcommand{\GRS}{{\mathrm{GRS}}}
\newcommand{\Hull}{{\mathrm{Hull}}}

\g@addto@macro\bfseries{\boldmath}

\makeatother

\usepackage{amsthm,amsmath,amssymb}
\usepackage{cite}
\usepackage{graphicx}
\usepackage{ulem}

\usepackage[colorlinks=true,citecolor=black,linkcolor=black,urlcolor=blue]{hyperref}

\theoremstyle{plain}
\newtheorem{theorem}{Theorem}[section]

\newtheorem{lemma}[theorem]{Lemma}
\newtheorem{corollary}[theorem]{Corollary}

\theoremstyle{definition}

\newtheorem{example}[theorem]{Example}

\theoremstyle{remark}
\newtheorem{remark}[theorem]{Remark}

\runningauthor{}

\date{}

\begin{document}
\begin{sloppypar}

\title{Several classes of Galois self-orthogonal MDS codes and related applications \thanks{The work of Yang Li, Yunfei Su and Shixin Zhu was supported by the National Natural Science Foundation of China under Grant
       Nos.U21A20428, 12171134, 61972126 and 62002093. The work of Shitao Li and Minjia Shi was supported by the National Natural Science Foundation of China under Grant No.12071001.}
\author{ Yang Li, Yunfei Su, Shixin Zhu\thanks{Corresponding author}, Shitao Li, Minjia Shi}
\thanks{ Yang Li, Yunfei Su and Shixin Zhu are with School of Mathematics, Hefei University of Technology, Hefei, China (email: yanglimath@163.com, suyunfei202208@163.com, zhushixinmath@hfut.edu.cn).
         Shitao Li and Minjia Shi are with School of Mathematical Sciences, Anhui University, Hefei, China (email: lishitao0216@163.com, smjwcl.good@163.com).}}

\date{}
    \maketitle

\begin{abstract}
  Let $q=p^h$ be a prime power and $e$ be an integer with $0\leq e\leq h-1$.
  $e$-Galois self-orthogonal codes are generalizations of Euclidean self-orthogonal codes ($e=0$) and Hermitian self-orthogonal codes ($e=\frac{h}{2}$ and $h$ is even).
  In this paper, we propose two general methods to construct $e$-Galois self-orthogonal (extended) generalized Reed-Solomon (GRS) codes.
  As a consequence, eight new classes of $e$-Galois self-orthogonal (extended) GRS codes with odd $q$ and $2e\mid h$ are obtained.
  Based on the Galois dual of a code, we also study its punctured and shortened codes.
  As applications, new $e'$-Galois self-orthogonal maximum distance separable (MDS) codes for all possible $e'$ satisfying $0\leq e'\leq h-1$,
  new $e$-Galois self-orthogonal MDS codes via the shortened codes, and new MDS codes with prescribed dimensional $e$-Galois hull via the punctured codes are derived. 
  Moreover, some new $\sqrt{q}$-ary quantum MDS codes with lengths greater than $\sqrt{q}+1$ and minimum distances greater than $\frac{\sqrt{q}}{2}+1$ are obtained.
  \end{abstract}
{\bf AMS Classification (MSC 2020)}: 94B05, 15B05, 12E10

\section{Introduction}\label{sec-introduction}

\subsection{Hulls and self-orthogonal maximum distance separable (MDS) codes}
Let $q=p^h$ be a prime power. Let $\mathbb{F}_q$ be the finite field with $q$ elements. 
An $[n,k,d]_q$ linear code $\C$ is a $k$-dimensional subspace of $\F_q^n$ with minimum distance $d$. 
$\C$ is called an maximum distance separable (MDS) code if the minimum distance $d$ reaches 
the Singleton bound: $d\leq n-k+1$, i.e., $d=n-k+1$. 
Since they have the maximum error correction capability when $n$ and $k$ are fixed, 
it always makes sense to construct MDS codes.
Particularly, (extended) generalized Reed-Solomon (GRS) codes are the most important families of MDS codes.

Let $\mathbf{a}=(a_1,a_2,\dots,a_n)$ and $\mathbf{b}=(b_1,b_2,\dots,b_n)\in \F_q^n$.
For each $0\leq e\leq h-1$, the $e$-Galois inner product of $\mathbf{a}$ and $\mathbf{b}$ is defined by
\begin{align}
   (\mathbf{a},\mathbf{b})_{e}=\sum _{i=1}^{n} a_i b_{i}^{p^{e}},
\end{align}
which was first introduced by Fan and Zhang \cite{RefJ5}.
The $e$-Galois inner product provides another tool to study algebraic structures of linear codes and 
generalizes the Euclidean inner product ($e=0$) and the Hermitian inner product ($e=\frac{h}{2}$ with even $h$). 
For an $[n,k,d]_q$ linear code $\C$, we can define its $e$-Galois dual code as
\begin{align}
  \C^{\bot_e}=\{ \mathbf{a} \in \F_q^n:\ (\mathbf{a},\mathbf{b})_e=0,\ {\rm for \ all} \ \mathbf{b}\in \C \}.
\end{align}

Denote the $e$-Galois hull of $\C$ by $\Hull_e(\C)=\C \cap \C^{\bot_e}$. Similarly, 
$\Hull_0(\C)$ is just the Euclidean hull and $\Hull_{\frac{h}{2}}(\C)$ coincides the Hermitian hull if $h$ is even. 
To classify finite projective planes, the $0$-Galois hull was first introduced by Assmus and Key \cite{Assmus}. In subsequent studies, one has found that
the $0$-Galois hull plays an important role in determining the equivalence of two linear codes and in computing the automorphism group of a linear code, 
and the $e$-Galois hull plays a significant role in calculating the parameter $c$ in an entanglement-assisted quantum error-correcting code (EAQECC) 
\cite{hull1,hull2,hull3,hull4,hull5}.


If $\Hull_e(\C)=\C$ (i.e., $\mathcal{C}\subseteq \mathcal{C}^{\bot_e}$), 
we call $\C$ an $e$-Galois self-orthogonal code, which contains the Euclidean self-orthogonal code 
and the Hermitian self-orthogonal code as two special cases. 
Self-orthogonal codes are of peculiar interest in theory and practice.
Many $t$-designs, optimal codes, linear complementary dual codes, quantum codes, EAQECCs, and so on 
have been derived from self-orthogonal codes 
(e.g., see \cite{RefJ1,RefJ4,RefJ7,RefJ8,RefJ12,RefJ11,RefJ13,RefJ18,RefJ30,RefJ31,CSS1,CSS2,quantum-Singleton-bound,5-design} and the references therein).
In particular, combining (extended) GRS codes and self-orthogonal codes, many Euclidean self-orthogonal and Hermitian self-orthogonal (extended) GRS codes were constructed
in \cite{RefJ60,RefJ11,RefJ40,RefJ16,RefJ26,RefJ50,RefJ70,Quantum-EGRS1,Quantum-EGRS2} and the references therein. Of course, it is also notable that
in these constructions fewer types of Hermitian self-orthogonal extended GRS codes were obtained. Nevertheless, by employing these self-orthogonal MDS codes,
a large number of quantum MDS codes with large minimum distance were further obtained. All achievements and applications mentioned above have made the
constructions of self-orthogonal MDS codes, especially self-orthogonal (extended) GRS codes, become an important and hot issue in coding theory.

\subsection{Our motivations and contributions}

In this paper, we propose general methods to construct $e$-Galois self-orthogonal (extended) GRS codes 
and study their applications. The motivations and contributions are summarized in three items. 

\begin{enumerate}
  \item [\rm 1)] Up to authors' knowledge, compared to the Euclidean and the Hermitian cases,
  there are few studies on Galois self-orthogonal MDS codes, especially, Galois self-orthogonal 
  (extended) GRS codes in the literature. 
  Specifically, in \cite{RefJ8,RefJ5,RefJ29,M.andC.,Fu2 Galois self-dual duadic constancyclic}, 
  sufficient conditions (some are also necessary) for (extended duadic) constacyclic codes and 
  skew multi-twisted codes over $\mathbb{F}_q$ being Galois self-orthogonal or Galois self-dual 
  were presented. And some of these Galois self-orthogonal codes are MDS. 
  In \cite{RefJ2,RefJ10,RefJ24,RefJ23}, with the aims of constructing EAQECCs 
  and MDS codes with hulls of arbitrary dimensions, 
  (extended) GRS codes were studied under the Galois inner product. 
  However, from these works, only five classes of Galois self-orthogonal GRS codes can be obtained and 
  no Galois self-orthogonal extended GRS codes can be derived.
  For ease of reference, we list them in Table \ref{tab:1}. 
  Naturally, a question in this topic is \textbf{how can we explicitly construct Galois self-orthogonal (extended) GRS codes?}
  Our investigation to answer this question leads to the following results: 
  \begin{itemize}
    \item Lemmas \ref{lem_Galois self-orthogonal GRS} and \ref{lem.Galois self-orthogonal EGRS} give 
    two new methods to construct Galois self-orthogonal (extended) GRS codes; 

    \item Based on these two general methods, combining different tools, we derive four new classes of 
     $e$-Galois self-orthogonal GRS codes in Theorems \ref{ConA.1}, \ref{ConB.1}, \ref{ConC.1} and \ref{ConD.1}, 
     and four new classes of $e$-Galois self-orthogonal extended GRS codes in 
     Theorems \ref{ConA.2}, \ref{ConB.2}, \ref{ConC.2} and \ref{ConD.2}, where $q=p^h$ is odd and $2e\mid h$. 
  \end{itemize}

  \item [\rm 2)] From \cite{RefJ23,RefJ24}, one can obtain more MDS codes of the same length 
  with $e'$-Galois hulls of arbitrary dimensions or linear codes of larger lengths with $e$-Galois 
  hulls of arbitrary dimensions from a given $e$-Galois self-orthogonal (extended) 
  GRS code. From these conclusions, some new Galois self-orthogonal codes of the same length or larger lengths 
  can be obtained. Then from a known Galois self-orthogonal MDS code, a natural question in this topic is  
\textbf{how can we explicitly construct an MDS code of shorter length with prescribed dimensional Galois hull?}
  Our investigation to answer this question leads to the following results: 
  \begin{itemize}
    \item Lemma \ref{lem-shorten-puncture}, Theorems \ref{thm-SO+Hull}, \ref{shorten} and 
    \ref{puncture} give some useful properties of the punctured and shortened codes under the 
    Galois inner product;  
    \item We derive many new MDS codes of shorter lengths with prescribed dimensional $e$-Galois hull
    and $e$-Galois self-orthogonal MDS codes of shorter length via the punctured and shortened codes 
    in Theorem \ref{th.EGRS_shortened};

    \item In addition, combining the results in \cite{RefJ23}, we further illustrate that many new 
    $e'$-Galois self-orthogonal MDS codes can be derived 
    in Theorem \ref{th.e'-Galois self-orthogonal codes} for all possible $e'$ satisfying $0\leq e'\leq h-1$. 
  \end{itemize}

  \item [\rm 3)] Note that once Galois self-orthogonal MDS codes are available, one can immediately get related Euclidean self-orthogonal MDS codes and Hermitian 
  self-orthogonal MDS codes by taking special cases. This gives us the third motivation and yields some quantum MDS codes in Theorem \ref{th.quantum MDS codes}.
\end{enumerate}

\newcommand{\tabincell}[2]{\begin{tabular}{@{}#1@{}}#2\end{tabular}}
\begin{table}[!htb]
\centering
\caption{Known Galois self-orthogonal GRS codes}
\label{tab:1}    
\begin{center}
  \resizebox{160mm}{20mm}{
  \begin{tabular}{ccccc}
    \hline
    Class &  Finite field & Code length & Dimension & Ref.\\
    \hline
  1 & $q=p^h$ is odd, $2e\mid h$ & $n=\frac{t(q-1)}{p^e-1}+1$, $1\leq t\leq p^e-1$ & $1\leq k\leq \lfloor \frac{p^e+n}{p^e+1} \rfloor$ & \cite{RefJ2} \\

  2 & $q=p^h$ is odd, $2e\mid h$ & \tabincell{c}{$n=\frac{r(q-1)}{\gcd(x_2,q-1)}+1$, $1\leq r\leq \frac{q-1}{\gcd(x_1,q-1)}$,\\ $(q-1)\mid {\rm{lcm}}(x_1,x_2)$ and $\frac{q-1}{p^e-1}\mid x_1$} & $1\leq k\leq \lfloor \frac{p^e+n}{p^e+1} \rfloor$ & \cite{RefJ2} \\

  3 & $q=p^h$ is odd, $2e\mid h$ & \tabincell{c}{$n=rm+1$, $1\leq r\leq \frac{p^e-1}{m_1}$,\\ $m_1=\frac{m}{\gcd(m,y)}$, $y=\frac{q-1}{p^e-1}$ and $m\mid (q-1)$} & $1\leq k\leq \lfloor \frac{p^e+n}{p^e+1} \rfloor$ & \cite{RefJ2} \\

  4 & $q=p^h$ is odd, $2e\mid h$ & \tabincell{c}{$n=tp^{aw}$, $a\mid e$\\ $1\leq t\leq p^a$, $1\leq w\leq \frac{h}{a}-1$} & $1\leq k\leq \lfloor \frac{p^e+n-1}{p^e+1} \rfloor$ & \cite{RefJ2} \\

  5 & $q=p^h$ is odd, $2^t\mid \frac{h}{m}$, $2^t=p^e+1$ & \tabincell{c}{$n=wp^{mz}$, $1\leq w\leq p^m$, $1\leq z\leq \frac{h}{m}-1$} & $1\leq k\leq \lfloor \frac{p^e+n-1}{p^e+1} \rfloor$ & \cite{RefJ23} \\
  \hline
  \end{tabular}}
\end{center}
\end{table}

\subsection{Organization of this paper}

After this introduction, in Section \ref{sec2}, we recall some necessary knowledge and present some results that are important for our topic. 
In Section \ref{sec3}, we construct eight new classes of $e$-Galois self-orthogonal (extended) GRS codes when $q=p^h$ is odd and $2e\mid h$.
In Section \ref{sec4}, we give many examples of general $e$-Galois self-orthogonal (extended) GRS codes. All of them are not Euclidean self-orthogonal or Hermitian self-orthogonal.
In Section \ref{sec_applications}, we propose some applications.  
More Galois self-orthogonal MDS codes, MDS codes with prescribed dimensional Galois hull and quantum MDS codes are obtained. 
In Section \ref{sec6}, we end this paper with some concluding remarks.

\section{Preliminaries}\label{sec2} 
\subsection{(Extended) generalized Reed-Solomon codes}
Throughout this paper, $e$ is an integer satisfying $0\leq e\leq h-1$ unless otherwise specified.
In this subsection, we introduce some basic knowledge and useful results on (extended) GRS codes.
For more details on (extended) GRS codes, readers are referred to \cite[Chapter 10]{RefJ100}.

Let $q=p^h$ be a prime power and $\mathbb{F}_q^n$ be an $n$-dimensional vector space over $\mathbb{F}_q$.
Choose two vectors $\mathbf{a}=(a_1,a_2,\dots,a_n)\in \mathbb{F}_{q}^n$
and $\mathbf{v}=(v_1,v_2,\dots,v_n)\in (\mathbb{F}_{q}^*)^n$, where $a_i\neq a_j$ for any $1\leq i\neq j\leq n\leq q$. We define
\begin{equation*}
    \GRS_k(\mathbf{a},\mathbf{v})=\{(v_1f(a_1),v_2f(a_2),\dots,v_nf(a_n)):f(x)\in \mathbb{F}_{q}[x], \deg(f(x))\leq k-1\}
\end{equation*}
as the GRS code with length $n$ and dimension $k$ associated to $\mathbf{a}$ and $\mathbf{v}$. Usually, the elements $a_1,a_2,\ldots,a_n$ are called the \underline{code locators} of
$\GRS_k(\mathbf{a}, \mathbf{v})$, and the elements $v_1,v_2,\ldots,v_n$ are called the \underline {column multipliers} of $\GRS_k(\mathbf{a}, \mathbf{v})$.

By adding an extra coordinate to $\GRS_k(\mathbf{a},\mathbf{v})$, the extended GRS code with length $n+1$ and dimension $k$ associated to $\mathbf{a}$ and $\mathbf{v}$, denoted
by $\GRS_k(\mathbf{a},\mathbf{v},\infty)$, is defined by
\begin{equation*}
    \GRS_k(\mathbf{a},\mathbf{v},\infty)=\{(v_1f(a_1),v_2f(a_2),\dots,v_nf(a_n),f_{k-1}):f(x)\in \mathbb{F}_{q}[x], \deg(f(x))\leq k-1\},
\end{equation*}
where $f_{k-1}$ is the coefficient of $x^{k-1}$ in $f(x)$. It is well known that GRS codes, extended GRS codes and
their dual codes are MDS codes.

As an important tool in our constructions, we denote
\begin{align}\label{equation.ui}
  u_i=\prod_{1\leq j\leq n,j\neq i}(a_i-a_j)^{-1},\ 1\leq i\leq n,
\end{align}
then it is clear that $u_i\neq 0$.

The following lemmas are useful for us to construct $e$-Galois self-orthogonal (extended) GRS codes.

\begin{lemma}{\rm(\cite[Lemma \Rmnum{3}.2]{RefJ2} and \cite[Lemma 2.3]{RefJ10})} \label{lem_all_1}
    Let $q=p^h$ with $p$ being an odd prime number and $E=\{x^{p^e+1}:\ x\in \mathbb{F}_q^*\}$ be a multiplicative subgroup of $\mathbb{F}_q^*$.
    Then $\mathbb{F}_{p^e}^*\subseteq E$ if and only if $2e\mid h$.
\end{lemma}

\begin{lemma}\label{lem_Galois hull GRS}{\rm(\cite[Proposition \Rmnum{2}.1]{RefJ2})}
    For $\mathbf{c}=(v_1f(a_1),v_2f(a_2),\dots,v_nf(a_n))\in \GRS_k(\mathbf{a},\mathbf{v})$, we have $\mathbf{c}\in \GRS_k(\mathbf{a},\mathbf{v})^{\bot_e}$
    if and only if there exists a polynomial $g(x)\in \mathbb{F}_q[x]$ with $\deg(g(x))\leq n-k-1$ such that
    \begin{align}
      \begin{split}
          (v_1^{p^e+1}f^{p^e}(a_1),v_2^{p^e+1}f^{p^e}(a_2),\dots,v_n^{p^e+1}f^{p^e}(a_n)) 
        =  (u_1g(a_1),u_2g(a_2),\dots,u_ng(a_n)).
      \end{split}
    \end{align}
\end{lemma}

\begin{lemma}\label{lem_Galois hull EGRS}{\rm(\cite[Proposition \Rmnum{2}.2]{RefJ2})}
     For $\mathbf{c}=(v_1f(a_1),v_2f(a_2),\dots,v_nf(a_n),f_{k-1})\in \GRS_k(\mathbf{a},\mathbf{v},\infty)$, we have
     $\mathbf{c}\in \GRS_k(\mathbf{a},\mathbf{v},\infty)^{\bot_e}$ if and only if there exists a polynomial $g(x)\in \mathbb{F}_q[x]$ with $\deg(g(x))\leq n-k$ such that
    \begin{align}
      \begin{split}
        & (v_1^{p^e+1}f^{p^e}(a_1),v_2^{p^e+1}f^{p^e}(a_2),\dots,v_n^{p^e+1}f^{p^e}(a_n),f_{k-1}^{p^e}) \\
        = & (u_1g(a_1),u_2g(a_2),\dots,u_ng(a_n),-g_{n-k}),
      \end{split}
    \end{align}
    where $g_{n-k}$ is the coefficient of $x^{n-k}$ in $g(x)$.
\end{lemma}

Now, we give the following two lemmas for later use.

\begin{lemma}\label{lem_Galois self-orthogonal GRS}
    For any codeword $\mathbf{c}=(v_1f(a_1),v_2f(a_2),\dots,v_nf(a_n))\in \GRS_k(\mathbf{a},\mathbf{v})$, if there exists a monic polynomial $h(x)\in \mathbb{F}_q[x]$ with
    $\deg(h(x))\leq n-(p^e+1)k+p^e-1$ such that
    \begin{equation}\label{EquationGalois self-orthogonal GRS}
    \lambda u_ih(a_i)=v_i^{p^e+1}, \ 1\leq i\leq n,
    \end{equation}
    where $\lambda\in \mathbb{F}_{q}^*$, then $\GRS_k(\mathbf{a},\mathbf{v})$ is an $e$-Galois self-orthogonal code of length $n$.
\end{lemma}
\begin{proof}
    According to the assumptions given, multiplying both sides of Equation (\ref{EquationGalois self-orthogonal GRS}) by $f^{p^e}(a_i)$, we have
  \begin{equation}\label{EquationGalois.self-orthogonal.2}
    \lambda u_ih(a_i)f^{p^e}(a_i)=v_i^{p^e+1}f^{p^e}(a_i),\ \ 1\leq i\leq n.
  \end{equation}

  Let $g(x)=\lambda h(x)f^{p^e}(x)$ and substitute $g(a_i)=\lambda h(a_i)f^{p^e}(a_i)$ into Equation (\ref{EquationGalois.self-orthogonal.2}). Then
  \begin{align*}
    \begin{split}
      \deg(g(x))  = \deg(h(x))+p^e\deg(f(x)) 
                  \leq n-(p^e+1)k+p^e-1+p^e(k-1) 
                  = n-k-1
    \end{split}
  \end{align*}
  and $u_ig(a_i)=v_i^{p^e+1}f^{p^e}(a_i)$ for $1\leq i\leq n$. Hence, according to Lemma \ref{lem_Galois hull GRS},
  we have $\mathbf{c}\in \GRS_k(\mathbf{a},\mathbf{v})^{\bot _e}$. That is $\GRS_k(\mathbf{a},\mathbf{v})\subseteq \GRS_k(\mathbf{a},\mathbf{v})^{\bot_e}$,
  which finishes the proof.
\end{proof}

\begin{lemma}\label{lem.Galois self-orthogonal EGRS}
For any codeword $\mathbf{c}=(v_1f(a_1),v_2f(a_2),\dots,v_nf(a_n),f_{k-1})\in \GRS_k(\mathbf{a},\mathbf{v},\infty)$,
if there exists a monic polynomial $h(x)\in \mathbb{F}_{q}[x]$ with $\deg(h(x))=p^e+n-(p^e+1)k$ such that
\begin{equation}\label{EquationGalois self-orthogonal EGRS}
  -u_ih(a_i)=v_i^{p^e+1}, \ 1\leq i\leq n,
\end{equation}
then $\GRS_k(\mathbf{a},\mathbf{v},\infty)$ is an $e$-Galois self-orthogonal code of length $n+1$.
\end{lemma}
\begin{proof}
  According to the assumptions given, multiplying both sides of Equation (\ref{EquationGalois self-orthogonal EGRS}) by $f^{p^e}(a_i)$, we have
  \begin{equation}\label{EquationGalois.self-orthogonal EGRS.2}
    -u_ih(a_i)f^{p^e}(a_i)=v_i^{p^e+1}f^{p^e}(a_i),\ \ 1\leq i\leq n.
  \end{equation}

Let $g(x)=-h(x)f^{p^e}(x)$ and substitute $g(a_i)=-h(a_i)f^{p^e}(a_i)$ into Equation (\ref{EquationGalois.self-orthogonal EGRS.2}). Then
\begin{align*}
  \begin{split}
    \deg(g(x))  = \deg(h(x))+p^e\deg(f(x)) 
                \leq p^e+n-(p^e+1)k+p^e(k-1) 
                = n-k
  \end{split}
\end{align*}
and $u_ig(a_i)=v_i^{p^e+1}f^{p^e}(a_i)$, $1\leq i\leq n$. We now claim that $f^{p^e}_{k-1}=-g_{n-k}$ in both cases $\deg(f(x))< k-1$ and $\deg(f(x))=k-1$.
\begin{itemize}
  \item \underline{$\textbf{Case 1:}$ $\deg(f(x))<k-1$.} If $\deg(f(x))<k-1$ (i.e., $\deg(f(x))\leq k-2$), then $f_{k-1}=0$ and
        \begin{equation*}
         \begin{split}
          \deg(g(x)) & = \deg(h(x))+p^e\deg(f(x)) \\
                     & \leq p^e+n-(p^e+1)k+p^e(k-2)\\
                     & = n-k-p^e\\
                     & < n-k.
         \end{split}
       \end{equation*}
        It follows that $f_{k-1}^{p^e}=0=-g_{n-k}$.

  \item \underline{$\textbf{Case 2:}$ $\deg(f(x))=k-1$.} If $\deg(f(x))=k-1$, then
        \begin{equation*}
          \begin{split}
            \deg(g(x)) & = \deg(h(x))+p^e\deg(f(x)) \\
                       & = p^e+n-(p^e+1)k+p^e(k-1) \\
                       & = n-k.
          \end{split}
        \end{equation*}
        It is easy to check that, in this case, the coefficients of the highest degree of $g(x)$ and $-h(x)f^{p^e}(x)$ are $g_{n-k}$ and $-f^{p^e}_{k-1}$, respectively.
        It follows from $g(x)=-h(x)f^{p^e}(x)$ that $f_{k-1}^{p^e}=-g_{n-k}$.
\end{itemize}

Combining the above results, from Lemma \ref{lem_Galois hull EGRS}, we have $\mathbf{c}\in \GRS_k(\mathbf{a},\mathbf{v}, \infty)^{\bot _e}$. That is
$\GRS_k(\mathbf{a},\mathbf{v},\infty)\subseteq \GRS_k(\mathbf{a},\mathbf{v},\infty)^{\bot_e}$, which finishes the proof.
\end{proof}

\subsection{The punctured codes and the shortened codes}
Let $q=p^h$ be a prime power. Let $\C$ be an $[n, k, d]_q$ linear code and $T$ be a set of $s$ coordinate positions in $\C$.
If we puncture $\C$ by deleting all the coordinates in $T$ in each codeword of $\C$, then it is well known that
the resulting code is still linear and has length $n - s$.
We denote the \underline{punctured code} by $\C^T$.
Consider the set $\C(T)$ of codewords which are $0$ on $T$. Then $\C(T)$ is a subcode of $\C$.
Puncturing $\C(T)$ on $T$ gives a code over $\F_q$ of length $n -s$ called the code shortened on $T$.
We denote the \underline{shortened code} by $\C_T$.

Based on the Euclidean dual and the Hermitian dual of $\C$, the dimensions of their punctured and shortened codes are calculated
in \cite[Theorem 1.5.7]{PS_Euclidean_dual} and \cite[Lemma 3]{PS_Hermitian_dual} under certain conditions. For our purpose, we adjust
the calculation to the general Galois form. To this end, we need the following lemmas.

Let $\sigma: \F_q \rightarrow \F_q, a \mapsto a^p$ be the Frobenius automorphism of $\F_q$. For any vector $\mathbf{a}=(a_1,a_2,\ldots,a_n)\in \F_q^n$, 
denote $\sigma(\mathbf{a})=(\sigma(a_1), \sigma(a_2),\ldots, \sigma(a_n))$.

\begin{lemma}{\rm(\cite[Proposition 2.2]{PX})}\label{prop.Galois dual}
  Let $q=p^h$ be a prime power. Then the following statements hold.
  \begin{enumerate}
    \item [\rm 1)] $\C^{\bot_e}=(\sigma^{h-e}(\C))^{\bot_0}=\sigma^{h-e}(\C^{\bot_0})$.
    \item [\rm 2)] $(\C^{\bot_e})^{\bot_f}=\sigma^{2h-e-f}(\C)$ for any $0\leq e,f\leq h-1$.
    In particular, $(\C^{\bot_0})^{\bot_0}=\C$ and $(\C^{\frac{h}{2}})^{\bot_\frac{h}{2}}=\C$ if $h$ is even.
  \end{enumerate}
\end{lemma}
\begin{remark}\label{rem.Galois dual}
  By \cite[Remark 3]{RefJ24}, in a natural way, one can treat the $h$-Galois inner as the $0$-Galois inner product.
  And hence, from Lemma \ref{prop.Galois dual}, we have $(\C^{\bot_e})^{\bot_{h-e}}=\C$ for any $0\leq e\leq h$.
\end{remark}

\begin{lemma}{\rm(\cite[Corollary 8]{RefJ24})}\label{lem.hull(C)=hull(C Galois dual)}
    Let $q=p^h$ be a prime power. Let $\C$ be an $[n,k,d]_q$ linear code. Then for any $0\leq e\leq h-1$, we have
    \begin{align}
      \dim(\Hull_e(\C))=\dim(\Hull_e(\C^{\bot_e})).
    \end{align}
\end{lemma}

\begin{lemma}\label{lem-shorten-puncture}
Let $q=p^h$ be a prime power. Let $\C$ be an $[n,k,d]_q$ code. Let $T$ be a set of $s$ coordinates. Then the following statements hold.
\begin{itemize}
  \item [\rm 1)] $(\C^{\perp_e})_T=(\C^T)^{\perp_e}$ and $(\C^{\perp_e})^T=(\C_T)^{\perp_e}$.
  \item [\rm 2)] If $s<d$, then $\C^T$ and $(\C^{\perp_e})_T$ have dimensions $k$ and $n- s- k$, respectively.
\end{itemize}
\end{lemma}
\begin{proof}
  1) By a method analogous to the proof in \cite[Theorem 1.5.7]{PS_Euclidean_dual}, we can easily get $(\C^{\perp_e})_T=(\C^T)^{\perp_e}$.
  From Remark \ref{rem.Galois dual}, replacing $\C$ by $\C^{\bot_{h-e}}$, we have $\C_T=((\C^{\bot_{h-e}})^T)^{\bot_e}$, where $0\leq e\leq h$.
  It follows that $(\C_T)^{\bot_{h-e}}=(\C^{\bot_{h-e}})^T$. Note that $0\leq e\leq h$, thus, we can replace $h-e$ by $e$ here, i.e.,
  $(\C^{\perp_e})^T=(\C_T)^{\perp_e}$. This completes the proof of 1).

  2) Similar to the proof in \cite[Theorem 1.5.7]{PS_Euclidean_dual} again, the desired result follows from 1) above.
\end{proof}

\begin{theorem}\label{thm-SO+Hull}
Let $q=p^h$ be a prime power. Let $\C$ be an $[n,k,d]_q$ code with $l$-dimensional $e$-Galois hull.
Then $\C$ is the direct sum of an $e$-Galois self-orthogonal $[n,s,d'\geq d]_q$ code and an $[n,k-s,d''\geq d]_q$ code with $(l-s)$-dimensional $e$-Galois hull for $1\leq s\leq l$.
\end{theorem}

\begin{proof}
Let $\{\alpha_1,\alpha_2,\ldots,\alpha_s,\alpha_{s+1},\ldots,\alpha_l,\alpha_{l+1},\ldots,\alpha_k\}$ be a basis of $\C$ such that
$\{\alpha_1,\alpha_2,\ldots,\alpha_s,\alpha_{s+1},\ldots,\alpha_l\}$ is a basis of $\Hull_e(\C)=\C\cap \C^{\perp_e}$.
Let $\C_1$ be a linear code generated by $\{\alpha_{1},\alpha_{2},\ldots,\alpha_{s}\}$. Then $\C_1$ is an $e$-Galois self-orthogonal $[n,s,d'\geq d]_q$ code.
Let $\C_2$ be a linear code generated by $\{\alpha_s,\alpha_{s+1},\ldots,\alpha_l,\alpha_{l+1},\ldots,\alpha_k\}$. Then $\dim(\Hull_e(\C_2))\geq l-s$ and $\C=\C_1\oplus \C_2$.
We claim that $\C_2$ is an $[n,k-s,d''\geq d]_q$ code with $(l-s)$-dimensional $e$-Galois hull.
Otherwise, we have $\dim(\Hull_e(\C))=\dim(\C_1)+\dim(\Hull_e(\C_2))>s+(l-s)=l$, which is a contradiction.
This completes the proof.
\end{proof}

Based on Lemma \ref{lem-shorten-puncture} and Theorem \ref{thm-SO+Hull}, we further obtain the following two theorems, 
which can be used to change parameters of codes.

\begin{theorem}\label{shorten}
  Let $q=p^h$ be a prime power. If there exists an $[n,k,d]_q$ linear code $\C$ with $l$-dimensional $e$-Galois hull.
  Then there exists a set of $s$ coordinate positions $T$ such that the shortened code $\C_T$ of $\C$ on $T$
  is an $[n-s,k-s,d^*\geq d]_q$ linear code with $(l-s)$-dimensional $e$-Galois hull for $1\leq s\leq l$.
  Moreover, if $\C$ is MDS, then $\C_T$ is also MDS.
\end{theorem}

\begin{proof}
  Let $\C$ be an $[n,k,d]_q$ linear code with $l$-dimensional $e$-Galois hull and a generator matrix $G$.
  Up to equivalence, we can set
  $$G=(I_k|A)=({\bf e}_{k,i}|{\bf a}_i)_{1\leq i\leq k},$$
  where ${\bf e}_{k,i}$ and ${\bf a}_i$ are the $i$-th rows of $I_k$ (the identity matrix) and $A$, respectively.
  Assume that $\{{\bf r}_j\}_{j=1}^l$ is a basis of $\Hull_e(\C)$ such that the first non-zero position of ${\bf r}_j$ is the $i_j$-th position. Without loss of generality, we may assume that
  $1\leq i_1< i_2< \cdots < i_l$. Then $i_l\leq k$. Otherwise, ${\bf r}_{i_l}=\textbf 0$, which is a contradiction.

  Let $T\subseteq  \{i_1,i_2,\ldots,i_l\}$ such that $|T|=s$. Let $J=\{1,2,\ldots,k\}\setminus T=\{j_1,j_2,\ldots,j_{k-s}\}$ such that $j_1<j_2<\cdots<j_{k-s}$.
  Then it can be checked that $\{{\bf r}_j\}_{j\in T}\cup \{({\bf e}_{k,i}|{\bf a}_i)\}_{i\in J}$ is a basis of $\C$. From Lemma \ref{thm-SO+Hull}, we know that the code $\C(T)$ generated by $\{({\bf e}_{k,i}|{\bf a}_i)\}_{i\in J}$ is an $[n,k-s,d^*\geq d]_q$ linear code with $(l-s)$-dimensional $e$-Galois hull.
  In addition, the generator matrix of the shortened code $\C_T$ on $T$ is
  $$G_T=({\bf e}_{k-s,i}|{\bf a}_{j_i})_{1\leq i\leq k-s},$$
  where ${\bf e}_{k-s,i}$ is the $i$-th row of $I_{k-s}$ and ${\bf a}_{j_i}$ is the ${j_i}$-th row of $A$ for $1\leq i\leq k-s$.
  Note that codewords of the code $\C(T)$ are 0 on $T$ and $\C_T$ is the punctured code of $\C(T)$ on $T$. And since $\C(T)$ is a linear code with $(l-s)$-dimensional $e$-Galois hull,
  $\C_T$ is an $[n-s,k-s,d^*\geq d]_q$ linear code with $(l-s)$-dimensional $e$-Galois hull.

  Moreover, if $\C$ is MDS, we have $d=n-k+1$. By the Singleton Bound, we know that $d^*\leq (n-s)-(k-s)+1=n-k+1$, which follows from $d^*\geq d=n-k+1$ that
  $d^*=n-k+1$. Hence, $\C_T$ is MDS. This completes the proof.
\end{proof}

\begin{theorem}\label{puncture}
    Let $q=p^h$ be a prime power.  Let $\C$ be an $[n,k,d]_q$ linear code with $l$-dimensional $e$-Galois hull.
    If $s<d$, then there exists a set of $s$ coordinate positions $T$ such that the punctured code $\C^T$ of $\C$
    on $T$ is an $[n-s,k,d^*\geq d-s]_q$ linear code with $(l-s)$-dimensional $e$-Galois hull for $1\leq s\leq l$.
    Moreover, if $\C$ is MDS, then $\C^T$ is also MDS.
  \end{theorem}
  
  \begin{proof}
    The parameters of the punctured code $\C^T$ of $\C$ on $T$ are obvious from Lemma \ref{lem-shorten-puncture} 2).
    We now prove $\dim(\Hull_e(\C^T))=l-s$.
    Since $\C$ is an $[n,k]_q$ linear code with $l$-dimensional $e$-Galois hull, it follows from Lemma \ref{lem.hull(C)=hull(C Galois dual)}
    that $\C^{\perp_e}$ is an $[n,n-k]_q$ linear code with $l$-dimensional $e$-Galois hull.
    According to Theorem \ref{shorten}, there exists a set of $s$ coordinate positions $T$ such that
    the shortened code $(\C^{\perp_e})_T$ of $\C^{\perp_e}$ on $T$ is an $[n-s,n-k-s]_q$ linear code with $(l-s)$-dimensional $e$-Galois hull.
    It follows from Lemma \ref{lem-shorten-puncture} 1) that $(\C^T)^{\perp_e}=(\C^{\perp_e})_T$ is an $[n-s,n-k-s]_q$ linear code with $(l-s)$-dimensional $e$-Galois hull.
    By Lemma \ref{lem.hull(C)=hull(C Galois dual)} again, we have $\dim(\Hull_e(\C^T))=\dim(\Hull_e((\C^T)^{\perp_e}))=l-s$.
    In summary, $\C^T$ is an $[n-s,k,d^*\geq d-s]_q$ linear code with $(l-s)$-dimensional $e$-Galois hull.
  
    Moreover, if $\C$ is MDS, then $d=n-k+1$. By the Singleton Bound, we know that $d^*\leq n-s-k+1$.
    And since $d^*\geq d-s=n-s-k+1$, we have $d^*=n-s-k+1$. Hence, $\C^T$ is MDS.
    This completes the proof.
\end{proof}

Taking $l=k$ in Theorems \ref{shorten} and \ref{puncture}, we immediately have the following corollaries. 
Additionally, for this case in Theorems \ref{shorten}, in order to make the shortened code meaningful in practice, we further assume $1\leq s\leq k-1$,
i.e., the dimension of the shortened code is guaranteed to be positive.

\begin{corollary}\label{coro.Galois SO MDS codes via shortened codes}
  Let $q=p^h$ be a prime power. If there exists an $[n,k,n-k+1]_q$ $e$-Galois self-orthogonal MDS code,
  then there also exists an $[n-s,k-s,n-k+1]_q$ $e$-Galois self-orthogonal MDS code for $1\leq s\leq k-1$.
\end{corollary}


\begin{corollary}\label{coro.MDS codes with prescribed Galois hull via shortened codes}
    Let $q=p^h$ be a prime power. If there exists an $[n,k,n-k+1]_q$ $e$-Galois self-orthogonal MDS code and $s<n-k+1$,
    then there exists an $[n-s,k,n-s-k+1]_q$ MDS code with $(k-s)$-dimensional $e$-Galois hull for $1\leq s\leq k$.
  \end{corollary}


\subsection{Quantum codes}
Let $q=p^h$ be a prime power. 
Let $[[N,K,D]]_q$ be a $q$-ary quantum code $\mathcal{Q}$ of length $N$ with
size $q^K$ and minimum distance $D$. Such a $q$-ary quantum code $\mathcal{Q}$
can be seen as a subspace of the Hilbert space $(\mathbb{C}^q)^{\otimes N}$ and
can correct no more than $\lfloor \frac{D-1}{2} \rfloor$ qubit-errors.
In line with classical codes, one of the central topics in the study of quantum
codes is to construct quantum codes with the largest error correction capability
for a given length $N$ and a given dimension $K$.

Independent of $q$, the parameters of a quantum code $\mathcal{Q}$ are constrained
by the so-called quantum Singleton bound as follows.
\begin{lemma}{\rm(Quantum Singleton bound \cite{quantum-Singleton-bound})}\label{lemma.quantum Singleton bound}
  Let $\mathcal{Q}$ be a $q$-ary $[[N,K,D]]_q$ quantum code, then
  \begin{align}\label{Equationquantum-Singleton-bound}
    2D\leq N-K+2.
  \end{align}
\end{lemma}

If a quantum code $\mathcal{Q}$ achieves the quantum Singleton bound, we call $\mathcal{Q}$ a quantum MDS code.
In \cite{quantum-Singleton-bound}, Ashikhmin et al. also presented the following Hermitian construction, which
allows one to construct quantum codes from classical codes.

\begin{lemma}{\rm(Hermitian construction \cite{quantum-Singleton-bound})}\label{lemma.Hermitian construction}
 Let $q=p^h$, where $h$ is even. If $\C$ is a $q$-ary $[n,k,d]_{q}$ linear code such that $\C^{\bot_{\frac{h}{2}}} \subseteq \C$,
 then there exists a $\sqrt{q}$-ary quantum code with parameters $[[n,2k-n,\geq d]]_{\sqrt{q}}$.
\end{lemma}

From Lemma \ref{lemma.Hermitian construction}, we have the following corollary.

\begin{corollary}\label{coro.Hermitian construction}
  Let $q=p^h$, where $h$ is even. If $\C$ is a $q$-ary $[n,k,d]_{q}$ Hermitian self-orthogonal MDS code,
  then there exists a $\sqrt{q}$-ary quantum MDS code with parameters $[[n,n-2k,k+1]]_{\sqrt{q}}$.
\end{corollary}

\section{Constructions}\label{sec3}
In this section, we use the two new methods introduced in Lemmas \ref{lem_Galois self-orthogonal GRS} and \ref{lem.Galois self-orthogonal EGRS} to construct
eight new classes of $e$-Galois self-orthogonal (extended) GRS codes when $q=p^h$ is odd and $2e\mid h$.

\subsection{Construction A via the trace mapping from $\mathbb{F}_q$ to $\mathbb{F}_{p^e}$}\label{ConA}

Let $q=p^h$ be an odd prime power and $h$, $t$, $e$ be integers with $e\mid h$. We fix $B\subseteq \mathbb{F}_{p^e}$ as an $\mathbb{F}_p$-linear subspace
with $|B|\geq t$. Consider the trace mapping:
\begin{equation}
    \begin{aligned}\text { Tr: }
    \mathbb{F}_{q} & \longrightarrow \mathbb{F}_{p^{e}}, \\
    x & \longmapsto x+x^{p^e}+\cdot \cdot \cdot +x^{p^{h-e}} .
    \end{aligned}
    \end{equation}
Let $b_1=0,b_2,\cdots,b_t$ be $t$ distinct elements of $B$. Define
\begin{align*}
  T_i=\{x\in \mathbb{F}_q:\ T_r(x)=b_i\},
\end{align*}
where $1\leq i\leq t$.
It is easy to check that $|T_i|=p^{h-e}$ and $T_i\cap T_j=\emptyset$ for any $1\leq i\neq j\leq t$.

Assume that $n=tp^{h-e}$ with $1\leq t\leq p^e$, and denote
\begin{align}\label{equaltion_ConA_ui}
  \mathcal{T}=\bigcup_{i=1}^{t}T_i=\{a_1,a_2,\ldots,a_n\}.
\end{align}

\begin{lemma}\label{lemma_ConA_ui}
  Let $a_i$ and $u_i$ be defined as in Equations (\ref{equaltion_ConA_ui}) and (\ref{equation.ui}), respectively.
  Given $1\leq i\leq n$, suppose $a_i\in T_{j_0}$ for some $1\leq j_0\leq t$. Then
  \begin{align}\label{Equationlemma_ConA_ui}
    u_i=\prod_{j\neq j_0,j=1}^{t}(Tr(a_i)-b_j)^{-1}.
  \end{align}
  And further, $u_i\in \mathbb{F}_{p^e}^*\subseteq E$ for $1\leq i\leq n$ when $2e\mid h$.
\end{lemma}
\begin{proof}
  By \cite[Theorem \Rmnum{3}.2]{RefJ10}, we can see that Equation (\ref{Equationlemma_ConA_ui}) holds and $u_i\in \mathbb{F}_{p^e}^*$ for $1\leq i\leq n$.
  Moreover, according to Lemma \ref{lem_all_1}, $u_i\in \mathbb{F}_{p^e}^*\subseteq E$ for $1\leq i\leq n$ when $2e\mid h$.
\end{proof}

\begin{theorem}\label{ConA.1}
  Let $q=p^h$ with $p$ being an odd prime number and $2e\mid h$. Assume that $n=tp^{h-e}$ with $1\leq t\leq p^e$. Then there exists
  an $[n,k,n-k+1]_q$ $e$-Galois self-orthogonal GRS code for $1\leq k\leq \lfloor \frac{p^e(tp^{h-2e}+1)-1}{p^e+1} \rfloor$.
\end{theorem}
\begin{proof}
  Let notations be the same as before. By Lemma \ref{lemma_ConA_ui}, $u_i\in \mathbb{F}_{p^e}^*\subseteq E$. Hence,
  there exists $v_i\in \mathbb{F}_q^*$ such that $u_i=v_i^{p^e+1}$ for $1\leq i\leq n$.
  Set $\mathbf{v}=(v_1,v_2,\dots,v_n)$.
  For $1\leq k\leq \lfloor \frac{p^e+n-1}{p^e+1} \rfloor=\lfloor \frac{p^e(tp^{h-2e}+1)-1}{p^e+1} \rfloor$, consider any codeword
  $\mathbf{c}=( v_1f(a_1), v_2f(a_2),\dots v_nf(a_n))\in \GRS_k(\mathbf{a},\mathbf{v})$ with $\deg(f(x))\leq k-1$. Let $\lambda=1$ and $h(x)=1$.
  Note that
  \begin{align*}
    \deg(h(x))=0\leq n+p^e-(p^e+1)k-1,
  \end{align*}
  then by Lemma \ref{lem_Galois self-orthogonal GRS}, we can deduce that $\GRS_k(\mathbf{a},\mathbf{v})$ is an $e$-Galois self-orthogonal GRS code of length $n$.
\end{proof}

\begin{theorem}\label{ConA.2}
    Let $q=p^h$ with $p$ being an odd prime number and $2e\mid h$. Assume that $n=tp^{h-e}$ with $1\leq t\leq p^e$. If $(p^e+1)\mid (tp^{h-2e}+1)$, then there exists an
    $[n+1,k,\frac{tq+p^e+2}{p^e+1}]_q$ $e$-Galois self-orthogonal extended GRS code, where $k=\frac{p^e(tp^{h-2e}+1)}{p^e+1}$.
\end{theorem}
\begin{proof}
    Let notations be the same as before. Since $u_i\in \mathbb{F}_{p^e}^*\subseteq E$ by Lemma \ref{lemma_ConA_ui},
    $-u_i\in \mathbb{F}_{p^e}^*\subseteq E$. Hence, there exists
    $v_i'\in \mathbb{F}_q^*$ such that $-u_i=(v_i')^{p^e+1}$ for $1\leq i\leq n$.
    Set $\mathbf{v'}=(v_1',v_2',\dots,v_{n}')$.
    Since $(p^e+1)\mid (tp^{h-2e}+1)$, we can set $k=\frac{p^e+n}{p^e+1}=\frac{p^e(tp^{h-2e}+1)}{p^e+1}$. Consider any codeword
    $\mathbf{c}=(v_1'f(a_1),v_2'f(a_2),\dots,v_n'f(a_n),f_{k-1}) \in \GRS_k(\mathbf{a},\mathbf{v'},\infty)$ with $\deg(f(x))\leq k-1$. Let $h(x)=1$. Note that
    \begin{align*}
        \deg(h(x))=0=p^e+n-(p^e+1)k,
    \end{align*}
    then by Lemma \ref{lem.Galois self-orthogonal EGRS}, we can deduce that $\GRS_k(\mathbf{a},\mathbf{v'},\infty)$ is an $e$-Galois self-orthogonal extended GRS code of length $n+1$.
    And by the Singleton bound, the desired result follows.
\end{proof}

\subsection{Construction B via the norm mapping from $\mathbb{F}^{*}_{q}$ to $\mathbb{F}^{*}_{p^{e}}$}\label{ConB}

Let $q=p^{h}$ be an odd prime power and $h$, $e$ be integers with $e\mid h$. Consider the norm mapping:
\begin{equation}
    \begin{aligned}\text { Norm: }
    \mathbb{F}_{q}^{*} & \longrightarrow \mathbb{F}_{p^{e}}^{*}, \\
    x & \longmapsto \prod_{i=0}^{\frac{h}{e}-1} x^{p^{i e}}=x^{\frac{q-1}{p^{e}-1}} .
    \end{aligned}
    \end{equation}
Label the elements of $\mathbb{F}^{*}_{p^{e}}$ as $b_{1}, b_{2},\dots, b_{p^{e}-1}$. Define
\begin{align*}
    N_{i}=\{x\in \mathbb{F}^{*}_{q}:\ Norm(x) = b_{i}\},
\end{align*}
where $1\leq i\leq p^e-1$. It is well known that $Norm$ is surjective.
Then $|N_{i}|=|ker(Norm)|=\frac{q-1}{p^{e}-1}$ can be deduced from \cite{RefJ2} and $N_{i}\cap N_{j} =\emptyset $ for any $1\leq i\neq j\leq p^e-1$.

Assume that $n=\frac{t(q-1)}{p^e-1}$ with $1\leq t\leq p^e-1$, and denote
\begin{equation}\label{equation.ConB.ai}
    \mathcal{N} =\bigcup ^{t}_{i=1} N_{i}=\{a_{1}, a_{2},\ldots, a_{n}\}.
\end{equation}

\begin{lemma}\label{lemma_ConB_ui}
    Let $a_i$ and $u_i$ be defined as in Equations (\ref{equation.ConB.ai}) and (\ref{equation.ui}).
    Given $1\leq i\leq n$, suppose $a_{i}\in N_{s}$ for some $1\leq s\leq t$. Assume that $e\mid h$. Then
    \begin{align}\label{equation.ConB.ui}
      u_{i}=a_{i}^{1-\frac{q-1}{p^{e}-1}}\prod _{1\leq s'\leq t,s'\neq s}(b_{s}-b_{s'})^{-1}.
    \end{align}
    And further, $a_i^{\frac{q-1}{p^e-1}-1}u_i\in \mathbb{F}^{*}_{p^{e}}\subseteq E$ for $1\leq i\leq n$ when $2e\mid h$.
\end{lemma}
\begin{proof}
    By \cite[Lemma \Rmnum{3}.1]{RefJ2}, we can see that Equation (\ref{equation.ConB.ui}) holds and $\prod _{1\leq s'\leq t,s'\neq s}(b_{s}-b_{s'})^{-1} \in \mathbb{F}^{*}_{p^{e}}$.
    Hence, $a_i^{\frac{q-1}{p^e-1}-1}u_i=\prod _{1\leq s'\leq t,s'\neq s}(b_{s}-b_{s'})^{-1} \in \mathbb{F}^{*}_{p^{e}}$.
    Moreover, according to Lemma \ref{lem_all_1}, $a_i^{\frac{q-1}{p^e-1}-1}u_i\in \mathbb{F}^{*}_{p^{e}}\subseteq E$ for $1\leq i\leq n$ when $2e\mid h$.
\end{proof}
\begin{theorem}\label{ConB.1}
    Let $q=p^{h}$ with $p$ being an odd prime number and $2e\mid h$. Assume that $n=\frac{t(q-1)}{p^{e}-1}$ with $1\leq t\leq p^{e}-1$. Then there exists an $[n,k,n-k+1]_q$ $e$-Galois
    self-orthogonal GRS code for $1\leq k\leq \frac{(t-1)(q-1)}{p^{2e}-1}$.
\end{theorem}
\begin{proof}
    Let notations be the same as before. By Lemma \ref{lemma_ConB_ui}, $a_i^{\frac{q-1}{p^e-1}-1}u_i\in \mathbb{F}^{*}_{p^{e}}\subseteq E$. Hence, there exists $v_{i}\in \mathbb{F}^{*}_{q}$ such that
    $a_i^{\frac{q-1}{p^e-1}-1}u_i=v_{i}^{p^{e}+1}$ for $1\leq i\leq n$. Set $\mathbf{v}=(v_1,v_2,\dots,v_n)$.
    Since $2e\mid h$, $(p^{2e}-1)\mid (t-1)(q-1)$. It follows that
    $$\lfloor \frac{p^e+n-\frac{q-1}{p^{e}-1}}{p^e+1} \rfloor=\lfloor \frac{p^e}{p^e+1}+\frac{(t-1)(q-1)}{p^{2e}-1}\rfloor=\frac{(t-1)(q-1)}{p^{2e}-1}.$$
    For $1\leq k\leq \lfloor \frac{p^e+n-\frac{q-1}{p^{e}-1}}{p^e+1} \rfloor=\frac{(t-1)(q-1)}{p^{2e}-1}$,
    consider any codeword $\mathbf{c}=(v_1f(a_1),v_2f(a_2),\dots,v_nf(a_n))\in \GRS_k(\mathbf{a},\mathbf{v})$ with $\deg(f(x))\leq k-1$.
    Let $\lambda =1$ and $h(x)=x^{\frac{q-1}{p^{e}-1}-1}\in \mathbb{F}_q[x]$. Note that
    \begin{align*}
      \deg(h(x))=\frac{q-1}{p^{e}-1}-1\leq p^e+n-(p^e+1)k-1,
    \end{align*}
    then by Lemma \ref{lem_Galois self-orthogonal GRS}, we can deduce that $\GRS_k(\mathbf{a},\mathbf{v})$ is an $e$-Galois self-orthogonal GRS code of length $n$.
  \end{proof}

   If we add $a_{n+1}=0$ to Equation (\ref{equation.ConB.ai}), a new class of $e$-Galois self-orthogonal extended GRS codes of length $n+2$ can be constructed. It is based on the
   following lemma, which can be derived from \cite[Theorem \Rmnum{3}.2]{RefJ2} and Lemma \ref{lem_all_1}, directly.

\begin{lemma}\label{lem.ConB.wi}
    Let notations be the same as before. Put $a_{n+1}=0$, then for any $1\leq i\leq n$,
    \begin{align*}
        \prod_{1\leq j\leq n+1, j\neq i}(a_{i}-a_{j})^{-1}=a_{i}^{-1}\prod_{1\leq j\leq n, j\neq i}(a_{i}-a_{j})^{-1}\in \mathbb{F}^{*}_{p^e},
    \end{align*}
    and for $i=n+1$,
    \begin{align*}
        \prod_{j=1}^{n}(a_{n+1}-a_{j})^{-1}=(-1)^{n+\frac{t(q-p^e)}{p^e-1}}\prod_{i=1}^t b_{i}^{-1}\in \mathbb{F}^{*}_{p^e}.
    \end{align*}
    And further, still denote $\prod_{1\leq i\leq n+1,j\neq i}(a_i-a_j)^{-1}$ by $u_i$, then $u_{i}\in \mathbb{F}^{*}_{p^e}\subseteq E$ for $1\leq i\leq n+1$ when $2e\mid h$.
\end{lemma}
\begin{theorem}\label{ConB.2}
    Let $q=p^{h}$ with $p$ being an odd prime number and $2e\mid h$. Assume that $n=\frac{t(q-1)}{p^{e}-1}$ with $1\leq t\leq p^{e}-1$.
    Then there exists an $[n+2,k,\frac{tp^e(q-1)}{p^{2e}-1}+2]_q$ $e$-Galois self-orthogonal extended GRS code where $k=\frac{t(q-1)}{p^{2e}-1}+1$.
\end{theorem}
\begin{proof}
    Let notations be the same as before. Since $u_{i}\in \mathbb{F}^{*}_{p^e}\subseteq E$ by Lemma \ref{lem.ConB.wi}, $-u_{i}\in \mathbb{F}_{p^{e}}^{*}\subseteq E$.
    Hence, there exists $v'_{i}\in \mathbb{F}^{*}_{q}$ such that $-u_{i}=(v'_{i})^{p^{e}+1}$ for $1\leq i\leq n+1$.
    Set $\mathbf{v'}=(v'_1,v'_2,\dots,v'_{n+1})$.
    Since $2e\mid h$, $(p^{2e}-1)\mid t(q-1)$, then we can set $k=\frac{p^e+n+1}{p^e+1}=\frac{t(q-1)}{p^{2e}-1}+1$. Consider any codeword
    $\mathbf{c}=(v_1'f(a_1),v_2'f(a_2),\dots,v_{n+1}'f(a_{n+1}),f_{k-1})\in \GRS_k(\mathbf{a},\mathbf{v'},\infty)$ with $\deg(f(x))\leq k-1$. Let $h(x)=1$. Note that
    \begin{align*}
      \deg(h(x))=0=p^e+n+1-(p^e+1)k,
    \end{align*}
    then by Lemma \ref{lem.Galois self-orthogonal EGRS}, we can deduce that $\GRS_k(\mathbf{a},\mathbf{v'},\infty)$ is an $e$-Galois
    self-orthogonal extended GRS code of length $n+2$. And by the Singleton bound, the desired result follows.
 \end{proof}

\subsection{Construction C via the direct product of two cyclic subgroups}\label{ConC}

Let $\omega$ be a primitive element of $\mathbb{F}_q$, where $q=p^h$ is an odd prime power. Denote by $\xi_1=\omega^{x_1}$ and $\xi_2=\omega^{x_2}$, where $x_1$ and $x_2$ are two positive integers.
Then ${\rm{ord}}(\xi_1)=\frac{q-1}{\gcd(q-1,x_1)}$ and ${\rm{ord}}(\xi_2)=\frac{q-1}{\gcd(q-1,x_2)}$.
Define $$R_i=\{\xi_1^i\xi_2^{j}:\ j=1,2,\dots,r_2\},$$ where $1\leq i\leq r_1$.
By \cite[Remark \Rmnum{3}.2]{RefJ2}, we know that $\langle \xi_1\rangle \otimes \langle \xi_2 \rangle$
is a subgroup of $\mathbb{F}_q^*$ with order ${\rm{ord}}(\xi_1)\cdot {\rm{ord}}(\xi_2)$ if $(q-1)\mid {\rm{lcm}}(x_1,x_2)$ and $R_s\cap R_t=\emptyset$ for any $1\leq s\neq t\leq r_1$.

Assume that $n=r_1r_2$ with $1\leq r_1\leq {\rm{ord}}(\xi_1)$, $r_2={\rm{ord}}(\xi_2)$, and denote
\begin{align}\label{equation.ConC_ai}
  \mathcal{R}=\bigcup_{i=1}^{r_1}R_i=\{a_1,a_2,\dots,a_n\}.
\end{align}
 Then we can derive the following lemma.

\begin{lemma}\label{lemma.ConC_ui}
  Let $a_i$ and $u_i$ be defined as in Equations (\ref{equation.ConC_ai}) and (\ref{equation.ui}), respectively.
  Given $1\leq i\leq n$, suppose $a_i\in R_s$ for some $1\leq s\leq r_1$. Assume that $(q-1)\mid {\rm{lcm}}(x_1,x_2)$ and $\gcd(x_2,q-1)\mid x_1(p^e-1)$ for
  two positive integers $x_1$ and $x_2$. Then
  \begin{align}\label{equation.ConC.ui}
    u_i=a_i\xi_1^{-sr_2}r_2^{-1}\prod_{1\leq s'\leq r_1,s'\neq s}(\xi_1^{sr_2}-\xi_1^{s'r_2})^{-1}.
  \end{align}
  And further, $a_i^{r_2-1}u_i\in \mathbb{F}_{p^e}^*\subseteq E $ for $1\leq i\leq n$ when $2e\mid h$.
\end{lemma}
\begin{proof}
  By \cite[Lemma \Rmnum{3}.4]{RefJ2}, we can get Equation (\ref{equation.ConC.ui}) directly.
  Now, we prove that $a_i^{r_2-1}u_i\in \mathbb{F}_{p^e}^*\subseteq E$ for $1\leq i\leq n$ when $2e\mid h$.
  Then by Lemma \ref{lem_all_1}, it is sufficient to prove $a_i^{r_2-1}u_i\in \mathbb{F}_{p^e}^{*}$ for $1\leq i\leq n$.
  Since $a_i\in R_s$ for some $1\leq s\leq r_1$, there exists an integer $t\in\{1,2,\dots,r_2\}$ such that $a_i=\xi _1^s\xi _2^t\in R_s$.
  Note that $r_2={\rm{ord}}(\xi_2)$, then $a_i^{r_2}=\xi_1^{sr_2}\xi_2^{tr_2}=\xi_1^{sr_2}$.
  Since $\gcd(x_2,q-1)\mid x_1(p^e-1)$, $\xi_1^{r_2}\in \mathbb{F}^*_{p^e}$.
  Hence, for each $1\leq i\leq n$,
  $$a_i^{r_2-1}u_i=r_2^{-1}\prod_{1\leq s'\leq r_1,s'\neq s}(\xi_1^{sr_2}-\xi_1^{s'r_2})^{-1} \in \mathbb{F}_{p^e}^*,$$
 which completes the proof.
\end{proof}

\begin{theorem}\label{ConC.1}
  Let $q=p^h$ with $p$ being an odd prime number and $2e\mid h$. Assume that $n=r_1r_2$ for each $1\leq r_1\leq \frac{q-1}{\gcd(q-1,x_1)}$ and $r_2=\frac{q-1}{\gcd(q-1,x_2)}$.
  If $(q-1)\mid {\rm{lcm}}(x_1,x_2)$ and $\gcd(x_2,q-1)\mid x_1(p^e-1)$ with two positive integers $x_1$ and $x_2$,
  then there exists an
  $[n,k,n-k+1]_q$ $e$-Galois self-orthogonal GRS code for $1\leq k\leq \lfloor \frac{p^e+r_2(r_1-1)}{p^e+1} \rfloor$.
\end{theorem}
\begin{proof}
  Let notations be the same as before. By Lemma \ref{lemma.ConC_ui}, $a_i^{r_2-1}u_i\in \mathbb{F}_{p^e}^*\subseteq E$ when $2e\mid h$. Hence, there exists $v_i\in \mathbb{F}_q^*$
  such that $a_i^{r_2-1}u_i=v_i^{p^e+1}$ for $1\leq i \leq n$. Set $\mathbf{v}=(v_1,v_2,\dots,v_n)$.
  For $1\leq k\leq \lfloor \frac{p^e+n-r_2}{p^e+1} \rfloor=\lfloor \frac{p^e+r_2(r_1-1)}{p^e+1} \rfloor$, consider any codeword
  $\mathbf{c}=(v_1f(a_1),v_2f(a_2),\dots,v_nf(a_n))\in \GRS_k(\mathbf{a},\mathbf{v})$ with
  $\deg(f(x))\leq k-1$. Let $\lambda=1$ and $h(x)=x^{r_2-1}\in \mathbb{F}_q[x]$. Note that
  \begin{align*}
    \deg(h(x))=r_2-1\leq p^e+n-(p^e+1)k-1,
  \end{align*}
  then by Lemma \ref{lem_Galois self-orthogonal GRS}, we can deduce that $\GRS_k(\mathbf{a},\mathbf{v})$ is an $e$-Galois self-orthogonal GRS code of length $n$.
\end{proof}

Next, by adding $a_{n+1}=0$ to Equation (\ref{equation.ConC_ai}), then a new family of $e$-Galois self-orthogonal extended GRS codes of length $n+2$ can be derived
based on the following lemma. Similarly, combining \cite[Theorem \Rmnum{3}.5]{RefJ2} and Lemma \ref{lem_all_1}, the lemma is direct, so we omit the proof again.

\begin{lemma}\label{lemma_ConC_wi}
  Let notations be the same as before. Put $a_{n+1}=0$, then for any $1\leq i\leq n$,
  \begin{align*}
    \prod_{1\leq j\leq n+1, j\neq i}(a_{i}-a_{j})^{-1}=a_{i}^{-1}\prod_{1\leq j\leq n, j\neq i}(a_{i}-a_{j})^{-1}\in \mathbb{F}^{*}_{p^e},
  \end{align*}
  and for $i=n+1$,
  \begin{align*}
    \prod_{j=1}^{n}(a_{n+1}-a_{j})^{-1}=(-1)^{n}\xi_1^{-\frac{r_1r_2(r_1+1)}{2}}\xi_2^{\frac{r_1r_2(r_2+1)}{2}}\in \mathbb{F}^{*}_{p^e}.
  \end{align*}
  Still denote $\prod_{1\leq i\leq n+1,j\neq i}(a_i-a_j)^{-1}$ by $u_i$. Then $u_{i}\in \mathbb{F}^{*}_{p^e}\subseteq E$ for $1\leq i\leq n+1$ when $2e\mid h$.
\end{lemma}

\begin{theorem}\label{ConC.2}
    Let $q=p^h$ with $p$ being an odd prime number and $2e\mid h$. Assume that $n=r_1r_2$ for each $1\leq r_1\leq \frac{q-1}{\gcd(q-1,x_1)}$ and $r_2=\frac{q-1}{\gcd(q-1,x_2)}$.
    If $(q-1)\mid {\rm{lcm}}(x_1,x_2)$, $\gcd(q-1,x_2)\mid x_1(p^e-1)$ and $(p^e+1)\mid n$ with two positive integers $x_1$ and $x_2$,
    then there exists an $[n+2,k,\frac{p^en}{p^e+1}+2]_q$ $e$-Galois self-orthogonal extended GRS code, where $k=\frac{n}{p^e+1}+1$.
\end{theorem}
\begin{proof}
  Let notations be the same as before. Since $u_i\in \mathbb{F}_{p^e}^*\subseteq E$ by Lemma \ref{lemma_ConC_wi}, $-u_i\in \mathbb{F}_{p^e}^*\subseteq E$. Hence, there exists $v'_i\in \mathbb{F}_q^*$
  such that $-u_i=(v'_i)^{p^e+1}$ for $1\leq i\leq n+1$. Set $\mathbf{v'}=(v'_1,v'_2,\dots,v'_{n+1})$.
  Since $(p^e+1)\mid n$, we can set $k=\frac{p^e+n+1}{p^e+1}=\frac{n}{p^e+1}+1$. Consider any codeword
  $\mathbf{c}=(v'_1f(a_1),v'_2f(a_2),\dots,v'_{n+1}f(a_{n+1}),f_{k-1}) \in \GRS_k(\mathbf{a},\mathbf{v'},\infty)$ with $\deg(f(x))\leq k-1$. Let $h(x)=1$. Note that
  \begin{align*}
    \deg(h(x))=0=p^e+n+1-(p^e+1)k,
\end{align*}
then by Lemma \ref{lem.Galois self-orthogonal EGRS}, we can deduce that $\GRS_k(\mathbf{a},\mathbf{v'},\infty)$ is an $e$-Galois self-orthogonal extended GRS code of length $n+2$.
And by the Singleton bound, the desired result follows.
\end{proof}

\subsection{Construction D via the coset decomposition of a cyclic group}\label{ConD}

Let $q=p^h$ be an odd prime power and $\omega$ be a primitive element of $\mathbb{F}_q$. For $e\mid h$, denote by $y=\frac{q-1}{p^e-1}$ and $m_1m_2=m\mid (q-1)$ with
$m_1=\frac{m}{\gcd(m,y)}$ and $m_2=\gcd(m, y)$. Consider $H=\langle \vartheta _1\rangle $, $G =\langle \vartheta _2\rangle $, where $\vartheta _1=\omega^{\frac{q-1}{m}}$,
$\vartheta _2=\omega^{\frac{y}{m_2}}$. Then ${\rm{ord}}(H)={\rm{ord}}(\vartheta_1)=m$ and ${\rm{ord}}(G)={\rm{ord}}(\vartheta_2)=(p^e-1)m_2$.

Assume that $n = rm$ with $1\leq r\leq \frac{p^e-1}{m_1}$, and denote
\begin{equation}\label{equation.ConD_ai}
    \mathcal{H}=\bigcup _{i=1}^{r}\eta _i H= \{a_1,a_2,\dots,a_n\},
\end{equation}
where $\eta _i$ is the left coset representative of $G/H$ for $i =1, 2,\dots, \frac{p^e-1}{m_1}$.
Readers can get further elaboration on the relation between $H$ and $G$ from \cite{RefJ2}.
Then the following lemma can be derived.

\begin{lemma}\label{lemma_ConD_ui}
  Let $a_i$ and $u_i$ be defined as in Equations (\ref{equation.ConD_ai}) and (\ref{equation.ui}), respectively. Given $1\leq i\leq n$, suppose $a_i\in \eta_sH$ for some $1\leq s\leq r$.
  Assume that $e\mid h$ and $m\mid (q-1)$. Then
  \begin{equation}\label{equation_ConD_ui}
    u_i= a_i\eta _{s}^{-m}m^{-1}\prod_{1\leq s'\leq r, s'\neq s}(\eta _{s}^{m}-\eta ^m_{s'})^{-1}.
  \end{equation}
  And further, $a_i^{m-1}u_i\in \mathbb{F}^{*}_{p^e}\subseteq E$ for $1\leq i\leq n$ when $2e\mid h$.
\end{lemma}
\begin{proof}
  By \cite[Lemma \Rmnum{3}.6]{RefJ2}, we can get Equation (\ref{equation_ConD_ui}) directly. Now, we prove that $a_i^{m-1}u_i\in \mathbb{F}^{*}_{p^e}\subseteq E$
  for $1\leq i\leq n$ when $2e\mid h$. Then by Lemma \ref{lem_all_1}, it is sufficient to prove $a_i^{m-1}u_i\in \mathbb{F}^{*}_{p^e}$ for $1\leq i\leq n$.
  Since $a_i\in \eta_sH$ for some $1\leq s\leq r$, there exists an integer $1\leq t\leq m$ such that $a_i=\eta _s\vartheta _{1}^{t}$. Note that ${\rm{ord}}(\vartheta_1)=m$, then
  $a_i^m=\eta_s^m \vartheta_{1}^{mt}=\eta_s^m$. For any $1\leq i\leq r$, note that $\eta _i$ is the left coset representative of $G/H$, then $\eta_i = \vartheta_2^j$ for some $1\leq j\leq (p^e-1)m_2$.
  Hence, $\eta_i^m=\vartheta_2^{jm}=\omega^{jm_1y}\in \mathbb{F}^*_{p^e}$. Moreover, we have
 \begin{align*}
        a_i^{m-1}u_i =m^{-1}\prod_{1\leq s'\leq r, s'\neq s}(\eta _{s}^{m}-\eta ^m_{s'})^{-1}\in \mathbb{F}^*_{p^e},
 \end{align*}
  which completes the proof.
\end{proof}

\begin{theorem}\label{ConD.1}
    Let $q=p^h$ with $p$ being an odd prime number and $2e\mid h$. Assume that $n=rm$ and $m\mid(q-1)$ for each $1\leq r\leq \frac{p^e-1}{m_1}$ with $m_1 =\frac{m}{\gcd(m,y)}$ for $y=\frac{q-1}{p^e-1}$.
    Then there exists an $[n, k,n-k+1]_q$ $e$-Galois self-orthogonal GRS code for $1\leq k\leq \lfloor\frac{p^e+m(r-1)}{p^e+1}\rfloor$.
\end{theorem}
\begin{proof}
    Let notations be the same as before. By Lemma \ref{lemma_ConD_ui}, $a_i^{m-1}u_i\in \mathbb{F}_{p^e}^*\subseteq E$ when $2e\mid h$. Hence, there exists $v_i\in \mathbb{F}_q^*$ such that $a_i^{m-1}u_i=v_i^{p^e+1}$
    for $1\leq i\leq n$. Set $\mathbf{v}=(v_1,v_2,\dots,v_n)$.
    For $1\leq k\leq \lfloor \frac{p^e+n-m}{p^e+1} \rfloor=\lfloor\frac{p^e+m(r-1)}{p^e+1}\rfloor$, consider any codeword
    $\mathbf{c}=(v_1f(a_1),v_2f(a_2),\dots,v_nf(a_n))\in \GRS_k(\mathbf{a},\mathbf{v})$
    with $\deg(f(x))\leq k-1$. Let $\lambda=1$ and $h(x)=x^{m-1} \in \mathbb{F}_q[x]$. Note that
    \begin{align*}
      \deg(h(x))=m-1\leq p^e+n-(p^e+1)k-1,
    \end{align*}
    then by Lemma \ref{lem_Galois self-orthogonal GRS}, we can deduce that $\GRS_k(\mathbf{a},\mathbf{v})$ is an $e$-Galois self-orthogonal GRS code of length $n$.
  \end{proof}

  Taking the same manners with Lemmas \ref{lem.ConB.wi} and \ref{lemma_ConC_wi}, we can derive the following lemma from \cite[Theorem \Rmnum{3}.8]{RefJ2} and Lemma \ref{lem_all_1}.

  \begin{lemma}\label{lemma_ConD_wi}
    Let notations be the same as before. Put $a_{n+1}=0$, then for any $1\leq i\leq n$,
    \begin{align*}
        \prod_{1\leq j\leq n+1, j\neq i}(a_{i}-a_{j})^{-1}=a_{i}^{-1}\prod_{1\leq j\leq n, j\neq i}(a_{i}-a_{j})^{-1}\in \mathbb{F}^{*}_{p^e},
    \end{align*}
    and for $i=n+1$,
    \begin{align*}
        \prod_{j=1}^{n}(a_{n+1}-a_{j})^{-1}=(-1)^{n}\vartheta _1^{-\frac{rm(m+1)}{2}}\prod_{i=1}^{r}\eta _{i}^{-m}\in \mathbb{F}^{*}_{p^e}.
    \end{align*}
    Still denote $\prod_{1\leq i\leq n+1,j\neq i}(a_i-a_j)^{-1}$ by $u_i$. Then $u_{i}\in \mathbb{F}^{*}_{p^e}\subseteq E$ for $1\leq i\leq n+1$ when $2e\mid h$.
  \end{lemma}

 Based on Lemma \ref{lemma_ConD_wi}, a new family of $e$-Galois self-orthogonal extended GRS codes 
 can be constructed in the same manner as Theorem \ref{ConC.2}.
 We give our construction in Theorem \ref{ConD.2} directly.

  \begin{theorem}\label{ConD.2}
    Let $q=p^h$ with $p$ being an odd prime number and $2e\mid h$. Assume that $n=rm$ and $m\mid(q-1)$ for each $1\leq r\leq \frac{p^e-1}{m_1}$ with $m_1 =\frac{m}{\gcd(m,y)}$ for $y=\frac{q-1}{p^e-1}$.
    If $(p^e+1)\mid n$, then there exists an $[n+2,k,\frac{p^en}{p^e+1}+2]_q$ $e$-Galois self-orthogonal extended GRS code for $k=\frac{n}{p^e+1}+1$.
  \end{theorem}

  We conclude this section with the following remark.

\begin{remark}\label{remark1} $\quad$
  \begin{enumerate}
    \item [\rm 1)] Comparing with Table \ref{tab:1}, we can easily conclude that our Constructions \ref{ConA} to \ref{ConD} are new. One can see further discussions
    about these code lengths in \cite[Section \Rmnum{5}]{RefJ2}. In addition, different tools are used in our constructions, which implies that our new methods
    introduced in Lemmas \ref{lem_Galois self-orthogonal GRS} and \ref{lem.Galois self-orthogonal EGRS} are flexible and uniform.

    \item [\rm 2)] In particular, for the case even $h$ and $e=\frac{h}{2}$, similar to Theorems \ref{ConC.1} and \ref{ConD.1}, two classes of classical Hermitian
    self-orthogonal GRS codes were constructed. Since Hermitian self-orthogonal codes are special cases of $e$-Galois self-orthogonal codes, these two classes of
    Hermitian self-orthogonal GRS codes are contained in our constructions.
  \end{enumerate}
\end{remark}

\section{Examples}\label{sec4}
In this section, we give some examples of general $e$-Galois self-orthogonal (extended) GRS codes, which are computed with the Magma software package \cite{magma}.
Recall that if $e=0$, $e$-Galois self-orthogonal codes are Euclidean self-orthogonal codes;
if $e=\frac{h}{2}$ with even $h$, $e$-Galois self-orthogonal codes are Hermitian self-orthogonal codes.
Hence, in the following examples, all codes listed are general $e$-Galois self-orthogonal codes rather than Euclidean self-orthogonal or Hermitian self-orthogonal codes.

\begin{example}\label{example1}
    Take $(p,h)=(3,8)$, $(p,h)=(5,8)$ and $(p,h)=(7,8)$ in Theorems \ref{ConA.1} and \ref{ConA.2}, then all cases satisfying $2e\mid h$ are $e=1$, $e=2$, and $e=4$.
    We list some general $e$-Galois self-orthogonal GRS codes over $\mathbb{F}_{3^{8}}$ and $e$-Galois self-orthogonal extended GRS codes over $\F_{3^{8}}$, $\F_{5^{8}}$ and
    $\mathbb{F}_{7^{8}}$ obtained from Theorems \ref{ConA.1} and \ref{ConA.2} in Tables \ref{tab:2} and \ref{tab:3}, respectively.
\end{example}


\begin{example}\label{examples2}
  Take $(p,h)=(3,8)$ again, we list some general $e$-Galois self-orthogonal (extended) GRS codes over $\mathbb{F}_{3^8}$ in Tables \ref{tab:4} and \ref{tab:5}.
\end{example}


\begin{example}\label{example3}

  Take $(p,h)=(3,4)$ and $(p,h)=(5,4)$, then all cases satisfying $2e\mid h$ are $e=1$ and $e=2$.
  Note that in this case, $2$-Galois self-orthogonal codes are just Hermitian self-orthogonal codes.
  For the general $e$-Galois self-orthogonal codes, we list some $1$-Galois self-orthogonal GRS codes
  over $\mathbb{F}_{5^4}$ and $1$-Galois self-orthogonal extended GRS codes over $\F_{3^4}$ and $\F_{5^4}$
  in Tables \ref{tab:6} and \ref{tab:7}, respectively.
\end{example}


\begin{example}\label{example4}
  Take $(p,h)=(5,8)$ again, we list some general $e$-Galois self-orthogonal (extended) GRS codes over $\mathbb{F}_{5^8}$ in Tables \ref{tab:8} and \ref{tab:10}.
\end{example}


\begin{table}[H]
    \caption{Some general $e$-Galois self-orthogonal GRS codes from Theorem \ref{ConA.1}}
       \label{tab:2}       
       \begin{center}
           \begin{tabular}{c|c|c|c}
            \hline
             $e$ & $t$ & \text{$e$-Galois self-orthogonal code} & $k$ \\\hline
             1 & 1 & $[2187,k,2188-k]_{3^8}$ & $1\leq k\leq 547$ \\
             1 & 2 & $[4374,k,4375-k]_{3^8}$ & $1\leq k\leq 1094$ \\
             1 & 3 & $[6561,k,6562-k]_{3^8}$ & $1\leq k\leq 1640$ \\

             2 & 1 & $[729,k,730-k]_{3^8}$ & $1\leq k\leq 73$ \\
             2 & 2 & $[1458,k,1459-k]_{3^8}$ & $1\leq k\leq 146$ \\
             2 & 4 & $[2916,k,2917-k]_{3^8}$ & $1\leq k\leq 292$ \\
             2 & 5 & $[3645,k,3646-k]_{3^8}$ & $1\leq k\leq 365$ \\
             2 & 7 & $[5103,k,5104-k]_{3^8}$ & $1\leq k\leq 511$ \\
            \hline
          \end{tabular}
       \end{center}
  \end{table}

  \begin{table}[H]
    \caption{Some general $e$-Galois self-orthogonal extended GRS codes from Theorem \ref{ConA.2}}
       \label{tab:3}       
       \begin{center}
           \begin{tabular}{c|c|c}
            \hline
             $e$ & $t$ & \text{$e$-Galois self-orthogonal code} \\\hline
              1 & 3 & $[6562,1641,4922]_{3^8}$ \\
              2 & 9 & $[6562,657,5906]_{3^8}$ \\
              1 & 5 & $[390626,65105,325522]_{5^8}$ \\
              2 & 25 & $[390626,15025,375602]_{5^8}$ \\
              1 & 7 & $[5764802,720601,5044202]_{7^8}$ \\
              2 & 49 & $[5764802,115297,5649506]_{7^8}$ \\
            \hline
          \end{tabular}
       \end{center}
  \end{table}


  \begin{table}[H]
    \caption{Some general $e$-Galois self-orthogonal GRS codes from Theorem \ref{ConB.1}}
       \label{tab:4}       
       \begin{center}
           \begin{tabular}{c|c|c|c}
            \hline
             $e$ & $t$ & \text{$e$-Galois self-orthogonal code} & $k$ \\\hline
             1 & 2 & $[6560,k,6561-k]_{3^8}$ & $1\leq k\leq 820$ \\

             2 & 2 & $[1640,k,1641-k]_{3^8}$ & $1\leq k\leq 82$ \\
             2 & 3 & $[2460,k,2461-k]_{3^8}$ & $1\leq k\leq 164$ \\
             2 & 4 & $[3280,k,3281-k]_{3^8}$ & $1\leq k\leq 246$ \\
             2 & 5 & $[4100,k,4101-k]_{3^8}$ & $1\leq k\leq 328$ \\
             2 & 6 & $[4920,k,4921-k]_{3^8}$ & $1\leq k\leq 410$ \\
             2 & 7 & $[5740,k,5741-k]_{3^8}$ & $1\leq k\leq 492$ \\
             2 & 8 & $[6560,k,6561-k]_{3^8}$ & $1\leq k\leq 574$ \\
            \hline
          \end{tabular}
       \end{center}
  \end{table}

  \begin{table}[H]
    \caption{Some general $e$-Galois self-orthogonal extended GRS codes from Theorem \ref{ConB.2}}
       \label{tab:5}       
       \begin{center}
        \begin{tabular}{c|c|c}
         \hline
          $e$ & $t$ & \text{$e$-Galois self-orthogonal code} \\\hline
           1 & 1 & $[3282,821,2462]_{3^8}$ \\

           2 & 1 & $[822,83,740]_{3^8}$ \\
           2 & 2 & $[1642,165,1478]_{3^8}$ \\
           2 & 3 & $[2462,247,2216]_{3^8}$ \\
           2 & 4 & $[3282,329,2954]_{3^8}$ \\
           2 & 5 & $[4102,411,3692]_{3^8}$ \\
           2 & 6 & $[4922,493,4430]_{3^8}$ \\
           2 & 7 & $[5742,575,5168]_{3^8}$ \\
         \hline
       \end{tabular}
    \end{center}
  \end{table}


\begin{table}[H]
  \caption{Some general $1$-Galois self-orthogonal GRS codes from Theorem \ref{ConC.1}}
     \label{tab:6}       
     \begin{center}
         \begin{tabular}{c|c|c}
          \hline
            $(x_1,x_2,n_1,n_2)$ & \text{$1$-Galois self-orthogonal code} & $k$ \\\hline
            $(156,208,4,3)$ & $[12,k,13-k]_{5^4}$ & $1\leq k\leq 2$ \\
            $(156,48,2,13)$ & $[26,k,27-k]_{5^4}$ & $1\leq k\leq 3$ \\
            $(156,48,3,13)$ & $[39,k,40-k]_{5^4}$ & $1\leq k\leq 5$ \\
            $(156,48,4,13)$ & $[52,k,53-k]_{5^4}$ & $1\leq k\leq 7$ \\
            $(156,16,2,39)$ & $[78,k,79-k]_{5^4}$ & $1\leq k\leq 7$ \\
            $(156,16,3,39)$ & $[117,k,118-k]_{5^4}$ & $1\leq k\leq 13$ \\
            $(156,112,4,39)$ & $[156,k,157-k]_{5^4}$ & $1\leq k\leq 20$ \\
          \hline
        \end{tabular}
     \end{center}
\end{table}

\begin{table}[H]
  \caption{Some general $1$-Galois self-orthogonal extended GRS codes from Theorem \ref{ConC.2}}
     \label{tab:7}       
     \begin{center}
         \begin{tabular}{c|c}
          \hline
            $(x_1,x_2,n_1,n_2)$ & \text{$1$-Galois self-orthogonal code} \\\hline
            $(720,780,1,4)$ & $[6,2,5]_{3^4}$ \\
            $(720,770,1,8)$ & $[10,3,8]_{3^4}$ \\
            $(720,775,1,16)$ & $[18,5,14]_{3^4}$ \\
            $(720,772,1,20)$ & $[22,6,17]_{3^4}$ \\
            $(780,416,2,3)$ & $[8,2,7]_{5^4}$ \\
            $(780,416,4,3)$ & $[14,3,12]_{5^4}$ \\
            $(624,754,1,24)$ & $[26,5,22]_{5^4}$ \\
            $(624,793,1,48)$ & $[50,9,42]_{5^4}$ \\
            $(624,712,1,78)$ & $[80,14,67]_{5^4}$ \\
          \hline
        \end{tabular}
     \end{center}
\end{table}


\begin{table}[H]
    \caption{Some general $e$-Galois self-orthogonal GRS codes from Theorem \ref{ConD.1}}
       \label{tab:8}       
       \begin{center}
           \begin{tabular}{c|c|c|c}
            \hline
             $e$ & $(m,r)$ & \text{$e$-Galois self-orthogonal code} & $k$ \\\hline
             1 & $(24,2)$ & $[48,k,49-k]_{5^8}$ & $1\leq k\leq 4$ \\
             1 & $(48,2)$ & $[96,k,97-k]_{5^8}$ & $1\leq k\leq 8$ \\
             1 & $(104,2)$ & $[208,k,209-k]_{5^8}$ & $1\leq k\leq 18$ \\
             1 & $(78,3)$ & $[234,k,235-k]_{5^8}$ & $1\leq k\leq 26$ \\

             2 & $(4,21)$ & $[84,k,85-k]_{5^8}$ & $1\leq k\leq 4$ \\
             2 & $(13,21)$ & $[273,k,274-k]_{5^8}$ & $1\leq k\leq 10$ \\
             2 & $(26,19)$ & $[494,k,495-k]_{5^8}$ & $1\leq k\leq 18$ \\
             2 & $(52,17)$ & $[884,k,885-k]_{5^8}$ & $1\leq k\leq 32$ \\
             \hline
          \end{tabular}
       \end{center}
  \end{table}

  \begin{table}[H]
    \caption{Some general $e$-Galois self-orthogonal extended GRS codes from Theorem \ref{ConD.2}}
       \label{tab:10}       
       \begin{center}
           \begin{tabular}{c|c|c}
            \hline
             $e$ & $(m,r)$ & \text{$e$-Galois self-orthogonal code}  \\\hline
             1 & $(2,3)$ & $[8,2,7]_{5^8}$  \\
             1 & $(12,3)$ & $[38,7,32]_{5^8}$  \\
             1 & $(24,4)$ & $[98,17,82]_{5^8}$  \\
             1 & $(78,3)$ & $[236,40,197]_{5^8}$  \\

             2 & $(13,8)$ & $[106,5,102]_{5^8}$  \\
             2 & $(13,22)$ & $[288,12,277]_{5^8}$  \\
             2 & $(52,13)$ & $[678,27,652]_{5^8}$  \\
             2 & $(52,22)$ & $[1146,45,1102]_{5^8}$  \\
             \hline
          \end{tabular}
       \end{center}
  \end{table}

\section{Applications}\label{sec_applications}

\subsection{Application to $e'$-Galois self-orthogonal MDS codes}

Note that $2e\mid h$ is required in our constructions. 
Let $q=p^h$ be a prime power and $e'$ be an integer satisfying $0\leq e'\leq h-1$, then combining with \cite{RefJ23}, 
we can apply the newly obtained $e$-Galois self-orthogonal GRS codes to derive new $e'$-Galois self-orthogonal MDS codes for all possible $e'$. 
To this end, we need the following lemma. 

\begin{lemma}{\rm(\cite[Theorems 5 1) and 8 1)]{RefJ23})}\label{lemma.CCDSYang}
  \begin{enumerate}
    \item  [\rm 1)] Let $q=p^h\geq 5$ and $1\leq e,e'\leq h-1$ such that $e=\gcd(e',h)$ and $\frac{h}{e}$ is even. If $\GRS_m(\mathbf{a},\mathbf{v})$ is an $[n,m,n-m+1]_q$ $e$-Galois
    self-orthogonal code, then for $1\leq k\leq \lfloor \frac{p^{e'}+n-1-\deg(m(x))}{p^{e'}+1} \rfloor$, there exists an $[n,k,n-k+1]_q$ $e'$-Galois self-orthogonal MDS code $\C$.

    \item [\rm 2)] Let $q=p^h$ be an odd prime power, where $h$ is even. Let $0\leq e'\leq h-1$ such that $\frac{h}{\gcd(e',h)}$ is odd. If $\GRS_m(\mathbf{a},\mathbf{v})$
    is an $[n,m,n-m+1]_q$ Hermitian self-orthogonal code, then for $1\leq k\leq \lfloor \frac{p^{e'}+n-1-\deg(m(x))}{p^{e'}+1} \rfloor$, there exists an $[n,k,n-k+1]_q$ $e'$-Galois
    self-orthogonal MDS code $\C$.
  \end{enumerate}
\end{lemma}

\begin{remark}\label{remark2}
Lemma \ref{lemma.CCDSYang} gives two special forms of Theorems 5 1) and 8 1) in \cite{RefJ23}, i.e., the case $\dim(\Hull_{e'}(\C))=l=k$.
Note that $\deg(m(x))$ is used to calculate the range of dimensions of new $e'$-Galois self-orthogonal MDS codes.
According to \cite[Remarks 1 and 3]{RefJ23}, $\deg(m(x))$ can be easily determined.
\end{remark}

\begin{theorem}\label{th.classify e'}
Let $h$ be an even positive integer and $0\leq e'\leq h-1$ be an integer. Let $S_1=\{e':\ \gcd(e',h)=e, 2e\mid h\}$ and $S_2=\{e':\ \gcd(e',h)=e, 2e\nmid h\}$
be two sets. Then for each $0\leq e'\leq h-1$, the following statements hold.
\begin{enumerate}
  \item [\rm 1)] $S_1\cap S_2=\emptyset$;

  \item [\rm 2)] $e'\in S_1$ or $e'\in S_2$.
\end{enumerate}
\end{theorem}
\begin{proof}
  1) It is trivial.

  2) Note that
  \begin{align*}
    \begin{split}
      & S_1=\{e':\ \gcd(e',h)=e, 2e\mid h\}=\{e':\ \frac{h}{\gcd(e',h)}\ \rm{is\ even}\}\ \\
      \rm{and}\ & S_2=\{e':\ \gcd(e',h)=e, 2e\nmid h\}=\{e':\ \frac{h}{\gcd(e',h)}\ \rm{is\ odd}\}.
    \end{split}
  \end{align*}

  Since $\gcd(e',h)\mid h$, $\frac{h}{\gcd(e',h)}$ is a positive integer. Hence, for each $0\leq e'\leq h-1$,
  $\frac{h}{\gcd(e',h)}$ is either even or odd, which follows that $e'$ is an element of $S_1$ or $S_2$. This completes the proof.
\end{proof}

\begin{theorem}\label{th.e'-Galois self-orthogonal codes}
  Let $q=p^h$ with $p$ being an odd prime number and $2e\mid h$. Let $\GRS_m(\mathbf{a},\mathbf{v})$ be an $[n,m,n-m+1]_q$ $e$-Galois self-orthogonal code
  derived from Theorems \ref{ConA.1}, \ref{ConB.1}, \ref{ConC.1} and \ref{ConD.1}. Then for $1\leq k\leq \lfloor \frac{p^{e'}+n-1-\deg(m(x))}{p^{e'}+1} \rfloor$,
  there exists an $[n,k,n-k+1]_q$ $e'$-Galois self-orthogonal MDS code $\C$ for $0\leq e'\leq h-1$.
\end{theorem}
\begin{proof}
    The condition $2e\mid h$ requires that $h$ is even. Hence, it is suitable to take $e=\frac{h}{2}$, i.e., $[n,m,n-m+1]_q$
    Hermitian self-orthogonal codes $\GRS_m(\mathbf{a},\mathbf{v})$ exist. And since $q=p^h$ is odd, we have $q\geq 9$, which
    satisfies the conditions of $q$ in Lemma \ref{lemma.CCDSYang}.

    Note that the relationships between $e'$ and $h$ described in Lemmas \ref{lemma.CCDSYang} 1) and \ref{lemma.CCDSYang} 2)
    correspond to the sets $S_1$ and $S_2$ in Theorem \ref{th.classify e'}, respectively.
    Therefore, on one hand, based on these $[n,m,n-m+1]_q$ $e$-Galois self-orthogonal codes $\GRS_m(\mathbf{a},\mathbf{v})$,
    it is easy to see that the desired $e'$-Galois self-orthogonal MDS codes can be deduced from Lemma \ref{lemma.CCDSYang} 1) for all $e'\in S_1$.
    On the other hand, based on these $[n,m,n-m+1]_q$ Hermitian self-orthogonal codes $\GRS_m(\mathbf{a},\mathbf{v})$, from Lemma \ref{lemma.CCDSYang} 2), one can derive
    the desired $e'$-Galois self-orthogonal MDS codes for all $e'\in S_2$. In summary, the desired result is proven.
\end{proof}

\begin{remark}\label{remark3} $\quad$
  \begin{enumerate}
    \item [\rm 1)] From Theorem \ref{th.e'-Galois self-orthogonal codes}, for all code lengths $n$ derived in Theorems \ref{ConA.1}, \ref{ConB.1}, \ref{ConC.1} and \ref{ConD.1},
    we can obtain $e'$-Galois self-orthogonal MDS codes of length $n$ for each $0\leq e'\leq h-1$ when $p$ is odd and $h$ is even. Recall that the dimension range
    of these $e'$-Galois self-orthogonal MDS codes is related to $\deg(m(x))$ and $\deg(m(x))$ can be easily determined.

    \item [\rm 2)] Note that many examples were listed in \cite[Examples 6 and 10]{RefJ23} and one can follow their calculation steps to obtain $e'$-Galois self-orthogonal MDS codes
    with explicit parameters. These calculations are similar, so we omit concrete examples in this subsection.

    \item [\rm 3)] From the proofs in Theorems \ref{th.classify e'} and \ref{th.e'-Galois self-orthogonal codes}, combining Theorem \ref{th.classify e'} 1), it is easy to see that
    Hermitian self-orthogonal GRS codes can not be used to derive $e'$-Galois self-orthogonal MDS codes, where $e'\in S_1$.
    It follows that new Galois self-orthogonal GRS codes constructed by Theorems \ref{ConC.1} and \ref{ConD.1} can not be obtained from Hermitian self-orthogonal GRS codes constructed
    in \cite{RefJ24}. Hence, from this aspect, we can also conclude that researches in this paper are more general and generalize previous work.
  \end{enumerate}
\end{remark}

\subsection{Application to new $e$-Galois self-orthogonal MDS codes, new MDS codes with prescribed dimensional $e$-Galois hull and new quantum MDS codes}

For the second application, we note the following facts: 
\begin{enumerate}
    \item [\rm 1)] On one hand, we note that the $e$-Galois self-orthogonal extended GRS codes constructed in Theorems \ref{ConA.2}, \ref{ConB.2}, \ref{ConC.2} and \ref{ConD.2} 
    have extra restrictions on lengths and dimensions; 
    \item [\rm 2)] On the other hand, we recall the following known results: 
    \begin{itemize}
        \item [\rm \rmnum{1})] According to \cite[Theorem 18]{RefJ24}, one can construct more $e$-Galois self-orthogonal codes of larger lengths 
        from a given $e$-Galois self-orthogonal code. However, they may not MDS; 
    
        \item [\rm \rmnum{2})] According to \cite[Theorems 5 and 8]{RefJ23}, from a known $e$-Galois self-orthogonal (extended) GRS codes, 
        one can easily obtain an MDS code of the same length with $e'$-Galois hulls of arbitrary dimensions; 
    
        \item [\rm \rmnum{3})] According to \cite[Corollary 27 and Theorem 30]{RefJ24}, codes of larger length with $e$-Galois hulls of arbitrary dimensions can be derived 
        from a given $e$-Galois self-orthogonal code.
    \end{itemize}
\end{enumerate}

Based on the facts above, in this subsection, we consider the application of the newly obtained $e$-Galois self-orthogonal extended GRS codes.  
Specifically, we use Corollary \ref{coro.Galois SO MDS codes via shortened codes} to obtain more $e$-Galois self-orthogonal MDS codes 
with shorter lengths and smaller dimensions and use Corollary \ref{coro.MDS codes with prescribed Galois hull via shortened codes} 
to obtain MDS codes of shorter lengths with prescribed dimensional $e$-Galois hull.  
Due to the MDS property and the new parameters, it is easy to see that these applications are new compared to the known results described above. 

We write these new results in Theorem \ref{th.EGRS_shortened}. According to Corollaries \ref{coro.Galois SO MDS codes via shortened codes} 
and \ref{coro.MDS codes with prescribed Galois hull via shortened codes}, these results are straight and we omit the proofs.  

\begin{theorem}\label{th.EGRS_shortened}
  Let $q=p^h$ with $p$ being an odd prime number and $2e\mid h$. Let $N$ and $K$ be two positive integers.
  If $N$ and $K$ satisfy one of the following four conditions:
  \begin{enumerate}
    \item [\rm \rmnum{1})] $N=tp^{h-e}+1$ with $1\leq t\leq p^e$, $(p^e+1)\mid (tp^{h-2e}+1)$ and $K=\frac{p^e(tp^{h-2e}+1)}{p^e+1}$;
    \item [\rm \rmnum{2})] $N=\frac{t(q-1)}{p^e-1}+2$ with $1\leq t\leq p^e-1$ and $K=\frac{t(q-1)}{p^{2e}-1}+1$;
    \item [\rm \rmnum{3})] $N=r_1r_2+2$ with $1\leq r_1\leq \frac{q-1}{\gcd(x_1,q-1)}$, $r_2=\frac{q-1}{\gcd(x_2,q-1)}$, $(q-1)\mid {\rm{lcm}}(x_1,x_2)$, $\gcd(x_2, q-1)\mid x_1(p^e-1)$, $(p^e+1)\mid n$ and $K=\frac{n}{p^e+1}+1$;
    \item [\rm \rmnum{4})] $N=rm+2$ with $1\leq r\leq \frac{p^e-1}{m_1}$, $m_1=\frac{m}{\gcd(m,y)}$, $y=\frac{q-1}{p^e-1}$, $m\mid (q-1)$, $(p^e+1)\mid n$ and $K=\frac{n}{p^e+1}+1$,
  \end{enumerate}
  then the following statements hold.
  \begin{enumerate}
    \item [\rm 1)] There exists an $[N,K,N-K+1]_q$ $e$-Galois self-orthogonal extended GRS code.  
    \item [\rm 2)] There exists an $[N-s,K-s,N-K+1]_q$ $e$-Galois self-orthogonal MDS code for $1\leq s\leq K-1$.
    \item [\rm 3)] If $s<N-K+1$, then there exists an $[N-s,K,N-s-K+1]_q$ MDS code with $(K-s)$-dimensional $e$-Galois hull for $1\leq s\leq K$.
  \end{enumerate}
\end{theorem}

\begin{example}\label{exam.shortened}
 From Table \ref{tab:7}, a $[22,6,17]_{3^4}$ $1$-Galois self-orthogonal extended GRS code and a $[26,5,22]_{5^4}$ $1$-Galois self-orthogonal extended GRS code exist.
 According to Theorem \ref{th.EGRS_shortened} 2), we know that
 $[21,5,17]_{3^4}$, $[20,4,17]_{3^4}$, $[19,3,17]_{3^4}$, $[18,2,17]_{3^4}$, $[17,1,17]_{3^4}$,
  $[25,4,22]_{5^4}$, $[24,3,22]_{5^4}$, $[23,2,22]_{5^4}$, $[22,1,22]_{5^4}$
   $1$-Galois self-orthogonal MDS codes also exist.
  Note that more examples can be similarly calculated.
  These facts illustrate that Theorem \ref{th.EGRS_shortened} 2) is flexible.
\end{example}

\begin{example}
 Again, from the $[22,6,17]_{3^4}$ $1$-Galois self-orthogonal extended GRS code and
 the $[26,5,22]_{5^4}$ $1$-Galois self-orthogonal extended GRS code discussed in Example \ref{exam.shortened}, 
 according to Theorem \ref{th.EGRS_shortened} 3), we can obtain
 $[21,6,16]_{3^4}$, $[20,6,15]_{3^4}$, $[19,6,14]_{3^4}$, $[18,6,13]_{3^4}$, $[17,6,12]_{3^4}$, and $[16,6,11]_{3^4}$ MDS codes with $5$, $4$, $3$, $2$, $1$, and $0$-dimensional $1$-Galois hull, respectively.
 We also can obtain $[25,5,21]_{5^4}$, $[24,5,20]_{5^4}$, $[23,5,19]_{5^4}$, $[22,5,18]_{5^4}$, and $[21,5,17]_{5^4}$ MDS codes with $4$, $3$, $2$, $1$, and $0$-dimensional $1$-Galois hull, respectively.
 Note that more examples can be similarly calculated.
 These facts illustrate that Theorem \ref{th.EGRS_shortened} 3) is flexible.
\end{example}

Additionally, to the best of our knowledge, a large number of known quantum MDS codes were constructed via Hermitian self-orthogonal GRS codes
in \cite{Quantum-EGRS1,Quantum-EGRS2,RefJ17,RefJ40,RefJ26,RefJ50,RefJ12,RefJ70,RefJ11} and the references therein.
However, a very small fraction was constructed via Hermitian self-orthogonal extended GRS codes in \cite{Quantum-EGRS1,Quantum-EGRS2,RefJ17,RefJ40}.
One of the reasons for this phenomenon is that, based on previous construction methods in \cite{RefJ2,FFLZ}, one is usually only able to
construct $[n,k,n-k+1]_q$ extended GRS codes with at most $(k-1)$-dimensional Hermitian hull.
Note that, in Theorem \ref{th.EGRS_shortened}, four classes of Hermitian self-orthogonal MDS codes can be obtained from Hermitian self-orthogonal extended 
GRS codes by taking $e = \frac{h}{2}$ with even $h$. 
For these reasons above, we can further obtain some new quantum MDS codes via these new Hermitian self-orthogonal MDS codes. 
Hence, we have the following theorem.

\begin{theorem}\label{th.quantum MDS codes}
  Let notations be the same as Theorem \ref{th.EGRS_shortened}. Then there exists
  an $[[N-s,N+s-2K,K+1-s]]_{\sqrt{q}}$ quantum MDS code for $0\leq s\leq K-1$.
\end{theorem}

\begin{example}
  In Tables \ref{tab:quantum MDS1} and \ref{tab:quantum MDS2}, we list some parameters of quantum MDS codes derived
  from Theorem \ref{th.quantum MDS codes}. Note that lengths of these $\sqrt{q}$-ary quantum MDS codes are greater
  than $\sqrt{q}+1$ and minimum distances are greater than $\frac{\sqrt{q}}{2}+1$.

  \begin{table}[H]
    \caption{Some quantum MDS codes derived from Theorem \ref{th.quantum MDS codes} satisfying the Condition \rmnum{2}) in Theorem \ref{th.EGRS_shortened}}
       \label{tab:quantum MDS1}       
       \begin{center}
           \begin{tabular}{c|c|c||c|c|c}
            \hline
             $t$ & $s$ & quantum MDS code & $t$ & $s$ & quantum MDS code \\\hline
             13 & 0 & $[[340,312,15]]_{25}$ & 14 & 1 & $[[365,337,15]]_{25}$ \\

             15 & 1 & $[[391,361,16]]_{25}$ & 15 & 2 & $[[390,362,15]]_{25}$ \\

             16 & 1 & $[[417,385,17]]_{25}$ & 16 & 3 & $[[415,387,15]]_{25}$ \\

             22 & 2 & $[[572,530,22]]_{25}$ & 22 & 4 & $[[570,532,20]]_{25}$ \\
             22 & 6 & $[[568,534,18]]_{25}$ & 22 & 8 & $[[566,536,16]]_{25}$ \\

             23 & 1 & $[[599,553,24]]_{25}$ & 23 & 3 & $[[597,555,22]]_{25}$ \\
             23 & 5 & $[[595,557,20]]_{25}$ & 23 & 7 & $[[593,559,18]]_{25}$ \\

            \hline
          \end{tabular}
       \end{center}
  \end{table}

  \begin{table}[H]
    \caption{Some quantum MDS codes derived from Theorem \ref{th.quantum MDS codes} satisfying the Condition \rmnum{4}) in Theorem \ref{th.EGRS_shortened}}
       \label{tab:quantum MDS2}       
       \begin{center}
           \begin{tabular}{c|c|c||c|c|c}
            \hline
             $(m,r)$ & $s$ & quantum MDS code & $t$ & $s$ & quantum MDS code \\\hline
             (50,25) & 0 & $[[1252,1200,27]]_{49}$ & (50,26) & 0 & $[[1302,1248,28]]_{49}$ \\

             (50,29) & 2 & $[[1450,1394,29]]_{49}$ & (50,29) & 4 & $[[1448,1396,27]]_{49}$ \\

             (50,36) & 1 & $[[1801,1729,37]]_{49}$ & (50,36) & 3 & $[[1799,1731,35]]_{49}$ \\
             (50,36) & 5 & $[[1797,1733,33]]_{49}$ & (50,36) & 7 & $[[1795,1735,31]]_{49}$ \\

             (50,47) & 8 & $[[2344,2264,41]]_{49}$ & (50,47) & 13 & $[[2339,2269,36]]_{49}$ \\
             (50,47) & 15 & $[[2337,2271,34]]_{49}$ & (50,47) & 20 & $[[2332,2276,29]]_{49}$ \\
            \hline
          \end{tabular}
       \end{center}
  \end{table}
\end{example}

\section{Summary and concluding remarks}\label{sec6}

The main contributions of this paper are to propose two new construction methods of $e$-Galois self-orthogonal (extended) GRS codes,
and these two methods lead to eight new classes of $e$-Galois self-orthogonal (extended) GRS codes when $q=p^h$ is odd and $2e\mid h$
(See Theorems \ref{ConA.1}, \ref{ConA.2}, \ref{ConB.1}, \ref{ConB.2}, \ref{ConC.1}, \ref{ConC.2}, \ref{ConD.1} and \ref{ConD.2}),
which enriches the researches on Galois self-orthogonal MDS codes.
Based on the general Galois dual of a code, we determine the dimensions of its punctured and shortened codes (See Lemma \ref{lem-shorten-puncture}).

We have also considered some related applications in this paper. We further illustrate that more $e'$-Galois self-orthogonal MDS 
codes can in fact be constructed via the newly obtained $e$-Galois self-orthogonal GRS codes, where $e'$ can take all possible values 
(See Theorem \ref{th.e'-Galois self-orthogonal codes}). 
In combination with the shortened and punctured codes, we construct many new $e$-Galois self-orthogonal 
MDS codes and new MDS codes with prescribed dimensional $e$-Galois hull (See Theorem \ref{th.EGRS_shortened}). 
Moreover, some new quantum MDS codes can be immediately obtained (See Theorem \ref{th.quantum MDS codes}).

For future research, it would be very interesting to further study the construction of $e$-Galois self-orthogonal (extended) GRS codes when $q=p^h$, where $h$ is odd.

\end{sloppypar}

\begin{thebibliography}{100}
\addtolength{\itemsep}{-1.5 em} 
\setlength{\itemsep}{-5pt}
\begin{footnotesize}

  \bibitem{Assmus} E.F. Assmus Jr, J.D. Key, Affine and projective planes, Discrete Math.  83(2-3), 161-187, (1990).
  \bibitem{quantum-Singleton-bound} A. Ashikhmin, E. Knill, Nonbinary quantum stabilizer codes, IEEE Trans. Inf. Theory 47(7), 3065-3072, (2001).
  \bibitem{5-design} E.F. Assmus Jr, H.F.  Mattson Jr, New 5-designs, J. Comb. Theory 6(2), 122-151 (1969).


  \bibitem{RefJ1} S. Bouyuklieva, Some optimal self-orthogonal and self-dual codes, Discret. Math. 287(1-3), 1-10, (2004).
  \bibitem{magma} W. Bosma, J. Cannon, C. Playoust, The Magma algebra system I: The user language, J. Symb. Comput. 24(3-4), 235-265, (1997).

  \bibitem{CSS1} A. Calderbank, P. Shor, Good quantum error-correcting codes exist, Phys. Rev. A, Gen. Phys.  54(2), 1098-1105, (1996).
  \bibitem{RefJ4} D. Crnković, R. Egan, A. ${\rm \check{S}}$vob, Constructing self-orthogonal and Hermitian self-orthogonal codes via weighing matrices and orbit matrices, Finite Fields Appl. 55, 64-77, (2019).
  \bibitem{RefJ2} M. Cao, MDS codes with Galois hulls of arbitrary dimensions and the related entanglement-assisted quantum error correction, IEEE Trans. Inf. Theory 67(12), 7964-7984, (2021).


  \bibitem{Quantum-EGRS2} W. Fang, F.W. Fu, Two new classes of quantum MDS codes, Finite Fields Appl. 53, 85-98, (2018).
  \bibitem{FFLZ} W. Fang, F.W. Fu, L. Li, S. Zhu, Euclidean and Hermitian hulls of MDS codes and their applications to EAQECCs, IEEE Trans. Inf. Theory 66(6), 3527-3537, (2019).
  \bibitem{RefJ7} X. Fang, M. Liu, J. Luo, New MDS Euclidean self-orthogonal codes, IEEE Trans. Inf. Theory 67(1), 130-137, (2020).
  \bibitem{RefJ10} X. Fang, R. Jin, J. Luo, W. Ma, New Galois hulls of GRS codes and application to EAQECCs, Cryptogr. Commun. 14(1), 145-159, (2022).
  \bibitem{RefJ8} Y. Fu, H. Liu, Galois self-orthogonal constacyclic codes over finite fields, Des. Codes Cryptogr. 90(11), 2703-2733, (2021).
  \bibitem{Fu2 Galois self-dual duadic constancyclic} Y. Fu, H. Liu, Galois self-dual extended duadic constacyclic codes. Discret. Math. 346(1), 113167, (2023) .
  \bibitem{RefJ5} Y. Fan, L. Zhang, Galois self-dual constacyclic codes, Des. Codes Cryptogr. 84(3), 473-492, (2017).



  \bibitem{RefJ12} C. Gan, C. Li, S. Mesnager, H. Qian, On hulls of some primitive BCH codes and self-orthogonal codes, IEEE Trans. Inf. Theory 67(10), 6442-6455, (2021).
  \bibitem{RefJ60} C. Galindo, F. Hernando, On the generalization of the construction of quantum codes from Hermitian self-orthogonal codes, Des. Codes Cryptogr. 90(5), 1103-1112, (2022).
  \bibitem{hull1} K. Guenda, S. Jitman, T.A. Gulliver, Constructions of good entanglement-assisted quantum error correcting codes, Des. Codes Cryptogr. 86(1), 121-136, (2018).
  \bibitem{RefJ11} G. Guo, R. Li, Y. Liu, Application of Hermitian self-orthogonal GRS codes to some quantum MDS codes, Finite Fields Appl. 76, 101901, (2021).


  \bibitem{RefJ13} F. Hernando, G. McGuire, F. Monserrat, J.J. Moyano-Fern��ndez, Quantum codes from a new construction of self-orthogonal algebraic geometry codes, Quantum Inf. Process. 19(4), 1-25, (2020).
  \bibitem{PS_Euclidean_dual} W.C. Huffman, V. Pless, Fundamentals of error-correcting codes, Cambridge University Press (2003).
  \bibitem{RefJ40} X. He, L. Xu, H. Chen, New $q$-ary quantum MDS codes with distances bigger than $\frac {q}{2}$, Quantum Inf. Process. 15(7), 2745-2758, (2016).


  \bibitem{RefJ16} L. Jin, S. Ling, C. Xing, Application of classical Hermitian self-orthogonal MDS codes to quantum MDS codes, IEEE Trans. Inf. Theory 56(9), 4735-4740, (2010).
  \bibitem{RefJ17} L. Jin, C. Xing, A construction of new quantum MDS codes, IEEE Trans. Inf. Theory 60(5), 2921-2925, (2014).


  \bibitem{RefJ18} J. Kim, Y. Kim, N. Lee, Embedding linear codes into self-orthogonal codes and their optimal minimum distances, IEEE Trans. Inf. Theory 67(6), 3701-3707, (2021).



  \bibitem{hull2} J. Leon, Computing automorphism groups of error correcting codes, IEEE Trans. Inf. Theory 28(3), 496-511, (1982).
  \bibitem{hull3} J.S. Leon, Permutation group algorithms based on partition, I: Theory and algorithms, J. Symb. Comput. 12(4-5), 533-583, (1991).
  \bibitem{PS_Hermitian_dual} G. Luo, M.F. Ezerman, S. Ling, Entanglement-assisted and subsystem quantum codes: New propagation rules and constructions, \url{https://arxiv.org/abs/2206.09782}, (2022).
  \bibitem{hull4} X. Liu, H. Liu, L. Yu, New EAQEC codes constructed from Galois LCD codes, Quantum Inf. Process. 19(1), 1-15, (2020).
  \bibitem{PX} H. Liu, X. Pan, Galois hulls of linear codes over finite fields, Des. Codes Cryptogr. 88(2), 241-255, (2020).
  \bibitem{Quantum-EGRS1} Z. Li, L. Xing, X. Wang, Quantum generalized Reed-Solomon codes: Unified framework for quantum maximum-distance-separable codes, Phys. Rev. A. 77(1), 012308, (2008).
  \bibitem{RefJ23} Y. Li, S. Zhu, P. Li, On MDS codes with Galois hulls of arbitrary dimensions, Cryptogr. Commun. \url{https://doi.org/10.1007/s12095-022-00621-3}, (2022).
  \bibitem{RefJ24} Y. Li, S. Zhu, On Galois hulls of linear codes and new entanglement-assisted quantum error-correcting codes, \url{https://arxiv.org/abs/2207.02535}, (2022).




  \bibitem{M.andC.} J. Mi, X. Cao, Constructing MDS Galois self-dual constacyclic codes over finite fields, Discret. Math. 344(6), 112388, (2021).
  \bibitem{RefJ100} F.J. MacWilliams, N.J.A. Sloane, The theory of error correcting codes, Elsevier, (1977).

  \bibitem{CSS2} A.M. Steane, Error correcting codes in quantum theory, Phys. Rev. Lett. 77(5), 793-797, (1996).
  \bibitem{hull5} N. Sendrier, Finding the permutation between equivalent codes: The support splitting algorithm, IEEE Trans. Inf. Theory 46(4), 1193-1203, (2000).
  \bibitem{RefJ29} A. Sharma, V. Chauhan, Skew multi-twisted codes over finite fields and their Galois duals, Finite Fields Appl. 59, 297-334, (2019).
  \bibitem{RefJ26} X. Shi, Q. Yue, Y. Chang, Some quantum MDS codes with large minimum distance from generalized Reed-Solomon codes, Cryptogr. Commun. 10(6), 1165-1182, (2018).
  \bibitem{RefJ50} X. Shi, Q. Yue, X. Zhu, Construction of some new quantum MDS codes, Finite Fields Appl. 46, 347-362, (2017).





  \bibitem{RefJ30} Y. Wu, Y. Lee, Binary LCD codes and self-orthogonal codes via simplicial complexes, IEEE Commun. Lett. 24(6), 1159-1162, (2020).



  \bibitem{RefJ70} T. Zhang, G. Ge, Quantum MDS codes with large minimum distance, Des. Codes Cryptogr. 83(3), 503-517, (2017).
  \bibitem{RefJ31} Z. Zhou, X. Li, C. Tang, C. Ding, Binary LCD codes and self-orthogonal codes from a generic construction, IEEE Trans. Inf. Theory 65(1), 16-27, (2018).

\end{footnotesize}

\end{thebibliography}
\end{document}